\documentclass[noinfoline,11pt]{imsart}

\RequirePackage[OT1]{fontenc}
\RequirePackage{amsthm,amsmath}
\RequirePackage[numbers]{natbib}
\RequirePackage[colorlinks,citecolor=blue,urlcolor=blue]{hyperref}

\usepackage{verbatim}

\allowdisplaybreaks

\usepackage{fullpage}

\usepackage{subfigure}
\usepackage{layout}

\usepackage{dsfont}
\usepackage[mathscr]{eucal}
\usepackage[toc,page]{appendix}
\usepackage{mathrsfs}
\usepackage{color}
\usepackage{pifont}
\usepackage{bm}
\usepackage{latexsym}
\usepackage{amsfonts}
\usepackage{amssymb}
\usepackage{epsfig}
\usepackage{graphicx}
\usepackage{multirow}

\newtheorem{theorem}{Theorem}
\newtheorem{proposition}{Proposition}

\newtheorem{remark}{Remark}

\startlocaldefs
\numberwithin{equation}{section}
\theoremstyle{plain}

\usepackage[OT1]{fontenc}
\usepackage{amsthm,amsmath,graphicx,latexsym,amssymb,array,mathabx,mathrsfs}
\usepackage{natbib}
\usepackage[colorlinks,citecolor=blue,urlcolor=blue]{hyperref}
\usepackage{adjustbox,blindtext}

\begin{document}

\begin{frontmatter}

\title{\large Regression with genuinely functional errors-in-covariates}

\runtitle{Functional Measurement Error}

\begin{aug}
\author{\fnms{Anirvan} \snm{Chakraborty}\ead[label=e1]{anirvan.chakraborty@epfl.ch}} \and
\author{\fnms{Victor M.} \snm{Panaretos}\ead[label=e2]{victor.panaretos@epfl.ch}}

\runauthor{A. Chakraborty \& V.M. Panaretos}

\affiliation{Ecole Polytechnique F\'ed\'erale de Lausanne}

\address{Institut de Math\'ematiques\\
Ecole Polytechnique F\'ed\'erale de Lausanne\\
\printead{e1}, \printead*{e2}}

\end{aug}

\begin{abstract} The contamination of covariates by measurement error is a classical and well-understood problem in multivariate regression, where it is well known that failing to account for this contamination can result in substantial bias in the parameter estimators. The nature and degree of this effect on statistical inference is also understood to crucially depend on the specific distributional properties of the measurement error in question. When dealing with functional covariates, measurement error has thus far been modelled as additive white noise over the observation grid. Such a setting implicitly assumes that the error arises purely at the discrete sampling stage, otherwise the model specification can only be viewed in a weak (stochastic differential equation) sense, white noise not being a second-order process. Departing from this simple distributional setting can have serious consequences for inference, similarly to the multivariate case, and current methodology will break down under such departures.  In this paper, we consider the situation where the additive measurement error is allowed to be a valid stochastic process.  We propose a novel estimator of the slope parameter in a functional linear model, for scalar as well as functional responses, in the presence of such a general specification of additive measurement error in the covariate. The proposed estimator is inspired by the multivariate regression calibration approach, but hinges on recent advances on matrix completion methods for functional data in order to handle the nontrivial (and unknown) measurement error covariance structure. The asymptotic properties of the proposed estimators are derived. We probe the performance of the proposed estimator of slope by means of numerical experiments and observe that it substantially improves upon the spectral truncation estimator based on the erroneous observations, i.e., ignoring measurement error.  We also investigate the behaviour of the estimators on a real dataset on hip and knee angle curves during a gait cycle.
\end{abstract}

\begin{keyword}[class=AMS]
\kwd[Primary ]{62M, 15A99}
\kwd[; secondary ]{62M15, 60G17}
\end{keyword}

\begin{keyword}
\kwd{grid sub-sampling}
\kwd{rank}
\kwd{regression calibration}
\kwd{slope parameter}
\end{keyword}

\end{frontmatter}

\tableofcontents

\newpage

\section{Introduction} \label{sec1}

\subsection*{Background}

Measurement error at the level of a covariate can adversely affect the quality of a regression estimator, if not properly accounted for, and this problem is by now well-studied and understood in the context of multivariate designs (see \cite{Full87} and \cite{CRSC06} for detailed expositions). The difficulties associated with measurement error are typically subtle, and can introduce identifiability issues, which may partly explain why \cite{CRSC06} refer to the inferential/statistical problems arising due to measurement errors in covariates as a \emph{double/triple whammy} of measurement error. Specifically, for multivariate linear regression, it is well known that classical additive measurement error typically results in bias in the estimator of slope in the direction of zero -- a phenomenon called attenuation to the null. The subtlety arises from the well-documented fact that  the specific distribution of the measurement error determines the effects it will have, and thus one has to properly account for this error distribution and devise appropriate methods (see p. 41 and also Ch. 3 of \cite{CRSC06}). 

\indent Functional regression considers stochastic processes as covariates, and is a contemporary topic with considerable activity (see \cite{marx1999}, \cite{Jame02}, \cite{RS05}, \cite{yao2005}, \cite{HH07}, \cite{CCKS07}, \cite{RO07}, \cite{CKS09}, \cite{YC10}, \cite{GBCCR11}, \cite{CM13}, \cite{ISSG15} to mention only a few). We refer to \cite{Morr15} and \cite{RGSO16} for informative and extensive reviews on such functional linear models. Some of these papers assume that the functional covariates are fully observed in the continuum and without any error, while others do indeed consider functional covariates measured on a grid and contaminated by measurement errors: the functional covariate $X$ is observed on a grid, say, $a \leq t_{1} < t_{2} < \ldots < t_{L} \leq b$, and the measurements obtained are $X_{j} = X(t_{j}) + \epsilon_{j}$. The measurement errors $\epsilon_{j}$ are then usually assumed i.i.d. zero mean random variables with common variance $\sigma^{2}$, i.e., the error structure is a classical, additive and homoscedastic one. 

This assumption allows one to circumvent measurement error via smoothing techniques, and two of the most prominent such approaches involve: (a) obtaining an estimate of the uncorrupted $X(\cdot)$ by spline smoothing of its corrupted version \cite{RS05}, and (b) smoothing the covariance of the corrupted observations $X_{j}$'s and then using this to compute finite rank approximations of the true $X$ using best linear prediction (see, e.g., \cite{YMW05}, \cite{hall2006properties}). Approach (a) is used when we observe functional data on a dense grid, while approach (b) is to be preferred if the observations are sparse, e.g., in the case of longitudinal covariates. These estimates may then used to carry out further statistical analysis. \cite{Jame02} considered natural cubic spline smoothing of the observed covariate $X_{j}$'s and imposed distributional assumptions to derive a maximum likelihood estimator of the discrete version of the parameters in generalized linear models. This can thus be considered as an adaptation of the likelihood-based approach for measurement error models (see, e.g., Ch. 8 in \cite{CRSC06}). \cite{yao2005} used the smoothed estimates of the covariance of $X$ developed in \cite{YMW05} to study estimators of slope parameter based on functional principal components for linear regression with longitudinal covariates. In their paper on a penalization criterion based approach for estimating the slope parameter, \cite{ISSG15} advocate the use of the general approaches (a) or (b), depending on the denseness/sparsity of the observation grid, when the covariates are measured with error. A different method for dealing with measurement error was studied by \cite{CCKS07}, who extended the total least squares approach (see, e.g., \cite{GVL80}, \cite{VHV91}) for multivariate error-in-variable model to the functional setting. \cite{CKS09} studied an estimator of the slope function based on smoothing splines after correcting for measurement errors using an approach similar to the total least squares method in \cite{CCKS07}. A two-step approach has been considered by \cite{GBCCR11}, where the authors first estimate $X(\cdot)$, then express these estimates and the slope parameter using a common spline basis, and then use a mixed effects model to estimate the slope parameter. These authors propose to estimate $X(\cdot)$ by first finding a smooth estimator of its covariance using a method in \cite{DCCP09} (which is similar to the approach of \cite{YMW05}) followed by spectral projection onto the first few components. \cite{GWC11} considered a variational Bayes approach for functional linear regression, which can be considered as an extension of the Bayesian method for measurement error models with multivariate data (see, e.g., Ch. 9 in \cite{CRSC06}). What is common to all these approaches is that they assume an i.i.d. structure for the measurement error, and this ingredient is crucial to their success.

An exception is a the recent Ph.D. thesis \citep{Cai2015}, in which the author studies an extension of the well-known SIMEX (simulation extrapolation, see, e.g., \cite{CS94} and Ch. 5 in \cite{CRSC06}) technique for multivariate regression with measurement errors to the functional setting with scalar response. Here the author \emph{does} also consider the case when the measurement errors may be heteroscedastic and/or autocorrelated, which appears to be the first such attempt in the literature. However, \cite{Cai2015} imposes a parametric specification of the error covariance, and, furthermore, the method and the asymptotic studies are only in the \emph{fixed} discretized grid setting. This effectively reduces the problem to the classical context of a multivariate version with parametrically specified measurement error covariance. A genuinely functional setup would need to consider a nonparametrically specified measurement error covariance structure and allow for theory and good practical performance as the observation grid becomes dense. 
 \\

\subsection*{Our Setup and Contributions}

Our purpose is to introduce the first genuinely functional method for handling regression with covariate contamination by measurement error that is not white noise. By `genuinely functional' we mean that the measurement error covariance should be nonparametric and that the observation grid should not be restricted to a finite size. Particularly in the functional data setting, it is natural to expect additional (dependence) structure in the measurement error. The error variances may not only be heteroscedastic, but there may well be propagation of error, so that errors associated with adjacent observation points need not be uncorrelated. One can then propose the more general contamination model $W(\cdot) = X(\cdot) + U(\cdot)$, where $U(\cdot)$ is a measurement error stochastic process that is a valid random element in the underlying function space, possessing a trace-class covariance, and $X$ and $U$ are uncorrelated. We call this a \emph{functional measurement error}, to be contrasted with the white noise setup where $U$ would be a generalised function interpretable only in weak SDE sense. 

Such a general specification of error inevitably leads to identifiability issues which cannot be alleviated without additional restrictions (see, e.g., \cite{OWY01}), not unlike the double/triple whammy mentioned earlier. This is perhaps a key reason why an i.i.d. error structure is typically assumed in FDA. Nevertheless, recent work by \cite{DP16} established a key nonparametric identifiability result that we will rely on, proving that the covariance operators $\mathscr{K}_{X}$ and $\mathscr{K}_{U}$ of $X$ and $U$, respectively, are identifiable provided the covariance kernel of $X$ is real analytic and that of $U$ is banded, i.e., there exists a $\delta > 0$ (called a \emph{bandwidth}) such that $\mathrm{cov}(U(t),U(s)) = 0$ if $|s-t| > \delta$ and $\mathscr{K}_{X}$ is analytic on an open set containing $|s-t| \leq \delta$. The required analyticity of the covariance kernel of $X$ is obviously ensured if $X$ is a process with a finite Karhunen-Lo\`eve expansion, i.e., it has finite rank, and its eigenfunctions are real analytic. Assuming a banded structure of $\mathscr{K}_{U}$ and analyticity of the eigenfunctions of $\mathscr{K}_{X}$ essentially translates to asking that the variability of the measurement error process is at a scale \emph{purely finer} than the scale of variability of the true covariate. 

\indent Given this identifiability result, we are naturally led to consider the following setup if we are to entertain a genuinely functional non-white measurement error framework. The true regression model is 
$$y = \alpha + \mathscr{B}X + \epsilon,$$ 
where $y$ is either a random scalar or function, $\mathscr{B}$ is a bounded linear operator from the underlying function space to the space of the response variable $y$, and $\epsilon$ is the error term. The true covariate $X$ is of some rank $r<\infty$, which beyond being finite can otherwise be unrestrictedly large. Instead of $X$, we only get to observe $W=X+U$, i.e. a version corrupted by a functional measurement error process $U$ with a banded but otherwise nonparametric covariance structure, as described in the previous paragraph. 

In the context of this model, propose a novel estimator of the slope parameter $\mathscr{B}$ in functional linear model, when the erroneous covariate $W$ is observed on a discrete grid of points. The estimator is motivate by the regression calibration approach (see, e.g., Ch. 4 in \cite{CRSC06}) for measurement error models in the multivariate setting. As part of this approach, one needs to construct a consistent estimator of the of the generalised inverse of the covariance operator of the true covariate $X$, on the basis of observing $W$, which is a severely ill-posed problem in itself (both due to the measurement error, as well as due to having to estimate an inverse) and constitutes the core challenge that one has to overcome. The crucial insight that allows us to overcome this challenge, is to use low-rank matrix completion to separate $X$ from $W$ (\`a la \cite{DP16}) combined with a novel grid-subsampling approach to estimate the true rank $r$ and allow for consistent inversion.

Our estimator is defined and studied in Section \ref{sec2} for both the scalar and the functional response case. In either case, the proposed estimators are shown to be consistent, and their rates of convergence are derived. Perhaps surprisingly, it is shown that $n^{-1/2}$ rates are entirely feasible (as the grid becomes dense), despite the convoluted nature of the problem. We discuss the practical implementation of the proposed estimators in Section \ref{subsec2-1}. In Section \ref{fqr}, we extend our methodology and asymptotic results to the case of quadratic functional regression with measurement error in the covariate. Simulation studies are conducted in Section \ref{sec3}, where the proposed estimator is compared with the spectral truncation estimator (see \cite{HH07}) constructed using the erroneous covariate. In subsection \ref{iid-errors}, we conduct simulations in the setting when the measurement errors are indeed i.i.d. over the observation grid and compare the proposed estimator with some of the other procedures that are designed for this setup. In Section \ref{sec4}, we discuss the situation when the true covariate $X$ is not truly finite rank but is essentially so, i.e., its eigenvalues decay fast but are not exactly eventually null. Here, we slightly modify the estimator of the covariance of $X$, but show that performance does not suffer. Finally, we illustrate the proposed estimator using a real data set  on hip and knee angle measurements in Section \ref{sec5}.

\section{Methodology, Theoretical Properties and Implementation} \label{sec2}

As discussed in the Introduction, the true regression model is taken to be
$$y = \alpha + \mathscr{B}X + \epsilon,$$
but we only get to observe the contaminated covariate $W = X + U$ instead of $X$ itself. Throughout this paper, we will assume that $E(X) = E(U) = 0$ for simplicity. We will also assume that $E(\epsilon \mid W) = E(\epsilon \mid X) = E(\epsilon \mid U) = 0$. Note that the assumption $E(\epsilon \mid W) = 0$ is automatically guaranteed if we assume that the measurement error is \emph{non-differential} (see, e.g., Sec. 2.5 in \cite{CRSC06}), i.e., the conditional distribution of $y$ given $(X,W)$ is the same as that of $y$ given $X$. This is in turn guaranteed if we assume that the measurement error $\epsilon$ in the response is independent of the measurement error $U$ in the covariate conditional on the true covariate $X$. Under non-differential measurement error, it follows that 
\begin{eqnarray*}
E(\epsilon \mid W) &=& E(y - \alpha - \mathscr{B}X \mid W) = E[E(y - \alpha - \mathscr{B}X \mid X,W) \mid W]\\
& =& E[E(y \mid X,W) - \alpha - \mathscr{B}X \mid W] = E[\{E(y \mid X) - \alpha - \mathscr{B}X\} \mid W] \\
&=& 0
\end{eqnarray*}
since the term inside the braces in the last expectation is zero. Consequently, 
$$E(y \mid W) = E(\alpha + \mathscr{B}X + \epsilon \mid W) = \alpha + \mathscr{B}E(X \mid W),$$ 
meaning that the regression of $y$ on $W$ is the same as that of $y$ on $E(X \mid W)$, i.e., one can fit the linear regression $y = \alpha + \mathscr{B}E(X \mid W) + \epsilon'$  instead of the model $y = \alpha + \mathscr{B}X + \epsilon$. These arguments provide the  justification for using of the regression calibration approach (see, e.g., Ch. 4 in \cite{CRSC06}). 

\subsection{Scalar-on-function model} \label{sec2-1}
\indent We first consider the case when the response is scalar. In this situation, the Riesz representation theorem implies that $$E(y \mid X) = \alpha + \langle X,\beta\rangle$$ for a unique $\beta$ (see, e.g., \cite{HE15}). 
In the multivariate setting, with random vectors $X$ and $U$ possessing full rank covariances, the least squares solution of $\beta$ using the regression calibration method is given by the solution of the equation $\mathrm{var}(E(X \mid W))\beta = \mathrm{cov}(y,E(X \mid W))$. Recall that $E(X \mid W) = K_{X}K_{W}^{-1}W$ and $\mathrm{var}(E(X \mid W)) = K_{X}K_{W}^{-1}K_{X}$, where $K_{W}$ and $K_{X}$ are the covariance matrices of $W$ and $X$, respectively. Straightforward algebra shows that the least squares solution is $$K_{X}^{-1}\mathrm{cov}(y,W).$$ 

It is thus observed that estimation of the slope necessitates the construction of a suitable estimator of the (generalised) inverse of the covariance of $X$. In the functional measurement error model proposed in the Introduction of this paper, it is assumed that the true covariate $X$ has finite rank $r$, i.e., the rank of $\mathscr{K}_{X} = r$. Even if we had a consistent estimator $\widehat{\mathscr{T}}$ of $\mathscr{K}_{X}$, this would not immediately guarantee that $\widehat{\mathscr{T}}^{-}$ would be a consistent estimator of  $\mathscr{K}_{X}^{-}$ (where for any compact self-adjoint operator $\mathscr{T}$, $\mathscr{T}^{-}$ denotes its Moore-Penrose generalized inverse; see, e.g., pp. 106-107 in \cite{HE15}). 
The map $\mathscr{T} \mapsto \mathscr{T}^{-}$ is continuous on the space of finite rank self-adjoint operators \emph{if and only if} for any sequence $\mathscr{T}_{n} \rightarrow \mathscr{T}$, we have $\mbox{rank}(\mathscr{T}_{n}) = \mbox{rank}(\mathscr{T})$ for all sufficiently large $n$. This says consistent estimation of $\mathscr{K}_{X}$ itself does not suffice, and we must be able to additionally consistently estimate the rank of $\mathscr{K}_{X}$ in order to be able to consistently estimate $\mathscr{K}_{X}^{-}$ accurately: the rank of the estimator $\widehat{\mathscr{T}}$ of $\mathscr{K}_{X}$ must accurately estimate the rank of $\mathscr{K}_{X}$. Note that despite $\mathscr{K}_X$ being finite rank, determining this finite rank is highly non-trivial in the presence of measurement error, since the potentially infinite rank of the measurement error process is confounded with the finite rank of the true covariate (in the absence of measurement error, estimating the rank would be a trivial problem once the number of observations and grid size exceeded the true rank $r$).

With these requirements in mind, we now develop our methodology.  Suppose that we observe each $W_{i}$, $1 \leq i \leq n$, over a grid of points $0 \leq t_{1} < t_{2} < \ldots, t_{L} \leq 1$, where $t_{j} \in I_{j,L}$ for each $j=1,2,\ldots,L$ with $\{I_{j,L} : 1 \leq j \leq L\}$ being a partition of $[0,1]$ into intervals of length $1/L$. We assume the grid nodes to be random. Define the discretely observed covariate vector as ${\bf W}_{i,L} = (W_{i}(t_{1}),W_{i}(t_{2}),\ldots,W_{i}(t_{L}))'$ for $1 \leq i \leq n$, and define the unobservable vectors ${\bf X}_{i,L}$'s and ${\bf U}_{i,L}$'s analogously. In this setting, \cite{DP16} proposed a consistent estimator of $\mathscr{K}_{X}$ which is defined as follows. Let $K_{X}$ be the covariance matrix of ${\bf X}_{1,L}$. The estimator $\widehat{K}_{X}$ of $K_{X}$ is obtained by minimizing 
\begin{equation}\label{DPestimator}
|||P_{L} \circ (\widehat{K}_{W} - \Theta)|||_{F}^{2} + \tau_{n}\mbox{rank}(\Theta).
\end{equation} 
Here, $\Theta$ ranges over the set of $L \times L$ positive definite matrices; $|||\cdot|||_{F}$ denotes the Frobenius norm of a matrix; $\widehat{K}_{W}$ is the empirical covariance matrix of the ${\bf W}_{i,L}$'s; $P_{L}$ is a matrix with $(i,j)$th entry $\mathbf{1}(|i-j| > \lceil L/4 \rceil)$; `$\circ$' denotes the element-wise (Hadamard) matrix product; and $\tau_{n}>0$ is a tuning parameter. The estimator of ${\mathscr{K}}_{X}$ is then defined as the integral operator associated with the kernel $\sum_{i=1}^{L}\sum_{j=1}^{L} \widehat{K}_{X}(i,j)\mathbf{1}\{(s,t) \in I_{i,L} \times I_{j,L}\}, s,t \in [0,1]$. In plain words, one estimates the covariance of $\bf X$ by a low-rank completion of the band-deleted covariance of $\bf W$ (the band deletion in principle decontaminating from the measurement error process). Under suitable conditions, this estimator can be shown to be consistent. However, as is observed from the asymptotic study in Theorem 3 in \cite{DP16}, one cannot guarantee that \emph{both} the estimator itself \emph{and} its rank will be consistent for their population counterparts (in fact, for dense grids,  $\mbox{rank}(\widehat{\mathscr{K}}_{X})$ may be inconsistent), which implies that we cannot naively use the same estimator in our context, as its generalised inverse may dramatically fail to be consistent for the true generalised inverse (as per the earlier discussion). 

 Since it is imperative that the rank be estimated as accurately as the operator itself in our case,  we introduce a modified two-step estimation procedure, estimating the rank separately from the operator itself: (1) In the first step, we minimise \eqref{DPestimator} on a subset of the grid evaluation points with the purpose of estimating the rank and not the covariance itself; (2) in the second step, we minimise $|||P_{L} \circ (\widehat{K}_{W} - \Theta)|||_{F}^{2}$ on the entire grid, over matrices $\Theta=\theta\theta'$, with $\theta\in\mathbb{R}^{L\times \hat{r}}$, where $\hat r$ comes from step (1). To explain the logic behind this strategy, we make the following technical observations: 
\begin{enumerate}
\item It follows from Theorem 2 in \cite{DP16} that the covariance matrices of ${\bf X}_{1,L}$ and ${\bf U}_{1,L}$ themselves (not just the continuum kernels they arise from) are identifiable provided that the grid size $L$ exceeds the critical value $4(r + 1)$ and that each of the $L$ grid evaluation points is sampled from a continuous distribution supported on its corresponding partition element $I_{j,L}$. Call such a grid an \emph{adequate grid}.

\item Therefore, the estimator $\widehat{r}_{L_{*}}$ of $r$ obtained by calculating the rank of the solution of  \eqref{DPestimator} on \textit{any} adequate grid of fixed size $L_*$ is consistent as $n \rightarrow \infty$ under the assumption that $\tau_{n} \rightarrow \infty$. Note that this would not yield a consistent estimator of $\mathscr{K}_{X}$ itself, though, since that would require the grid size to grow to infinity as well. See Theorem 3 in \cite{DP16}.

\item So if we can subsample our original grid to obtain an \emph{adequate} subgrid and keep the size of this subgrid fixed as $n\rightarrow\infty$, we will be able to estimate the rank consistently. See Section \ref{subsec2-1} for more details on how to choose this adequate grid. 

\item Now reasoning behind the two-step method becomes clear: we use the subgrid of fixed size $L_*$ in order to get a consistent estimator of the rank $\widehat{r}_{L_{*}}$, and we use the complete grid (whose size $L$ is in principle growing as $n$ grows) to obtain a consistent estimator of $\mathscr{K}_X$.
\end{enumerate}

In summary, the proposed estimation procedure is as follows.
\begin{itemize}
\item[Step 1:] Extract an \emph{adequate} subgrid of resolution $L_{*}<L$ as described above to get an estimator $\widehat{r}_{L_{*}}$ of the rank $r$.  Specifically, $\widehat{r}_{L_{*}}$ is defined as the rank of the minimizer of \eqref{DPestimator} based on the subsampled grid of size $L_*$. 

\item[Step 2:] Next, use the full grid of resolution $L$ to find an intermediate estimator of $K_{X}$ with rank equal to $\widehat{r}_{L_{*}}$ (obtained in Step 1) as follows:
\begin{equation}\label{modified}
 \widehat{K}_{X} := \mbox{arg}\min_{\Theta : \mbox{rank}(\Theta) = \widehat{r}_{L_{*}}} ||P_{L} \circ (\widehat{K}_{W} - \Theta)|||_{F}^{2},
 \end{equation}
where $\Theta$ is a $L \times L$ positive definite matrix, and the $(i,j)$th element of $P_{L}$ equals $\mathbf{1}(|i - j| > \lceil L/4\rceil)$. Construct the estimator  $\widehat{k}_{X}$ of the kernel $k_X$ of $\mathscr{K}_X$ as follows:
$$ \widehat{k}_{X}(s,t) = \sum_{i=1}^{L}\sum_{j=1}^{L} \widehat{K}_{X}(i,j)\mathbf{1}\{(s,t) \in I_{i,L} \times I_{j,L}\}, ~s,t \in [0,1].$$
Let $\widehat{\lambda}_{j}$ and $\widehat{\eta}_{j}$ denote the eigenvalues and the eigenfunctions of the integral operator associated with the kernel $\widehat{k}_{X}(s,t)$.

\item[Step 3:] Define the estimator $\widehat{\mathscr{K}}_X$ of $\mathscr{K}_{X}$ as
$$ \widehat{\mathscr{K}}_{X} = \sum_{j=1}^{\widehat{r}_{L_{*}}} \widehat{\lambda}_{j} (\widehat{\eta}_{j} \otimes \widehat{\eta}_{j}),$$
and use its Moore-Penrose inverse 
$$ \widehat{\mathscr{K}}_{X}^{-} = \sum_{j=1}^{\widehat{r}_{L_{*}}} \widehat{\lambda}_{j}^{-1} (\widehat{\eta}_{j} \otimes \widehat{\eta}_{j})$$
as the estimator of $\mathscr{K}_X^{-}$.
\end{itemize}

\begin{remark}
Strictly speaking, we should speak of \emph{a} minimiser rather than \emph{the} minimiser of \eqref{DPestimator} or of \eqref{modified}, since the minima of these objectives may not be unique. However, it can be shown that for all $n$ sufficiently large both objectives \emph{will} have a unique minimum, which is why we do not insist on making the pedantic distinction.\\
\end{remark}

\noindent The regression calibration estimator of the slope parameter $\beta$ in the functional linear regression setting is now defined as
$$\widehat{\beta}_{rc} = \widehat{\mathscr{K}}_{X}^{-}\widehat{C}_{y,W},$$
where $\widehat{C}_{y,W}$ is the empirical covariance between the $y_{i}$'s and the $W_{i}$'s.\\

\noindent The following theorem provides the asymptotic behaviours of $\widehat{\mathscr{K}}_{X}$, $\widehat{\mathscr{K}}_{X}^{-}$ and $\widehat{\beta}_{rc}$.\\

\begin{theorem}  \label{thm1}
Suppose that $E(||X||^{4}) < \infty$, $E(||U||^{4}) < \infty$, $E(\epsilon^{2}) < \infty$, $\delta < 1/4$, and let $L_{*} \geq 4(r+1)$ be a fixed integer. Suppose that $\tau_{n} \rightarrow 0$, $n\tau_{n} \rightarrow \infty$ and $L^{-2} = O(n^{-1})$ as $n \rightarrow \infty$. Then, the following hold as $n \rightarrow \infty$. \\
\begin{itemize}
\item[(a)] $P(\widehat{r}_{L_{*}} = r) \rightarrow 1$. \\
\item[(b)] $|||\widehat{\mathscr{K}}_{X}^{-} - \mathscr{K}_{X}^{-}|||_{HS}=O_{P}(n^{-1/2})$ and $|||\widehat{\mathscr{K}}_{X} - \mathscr{K}_{X}|||_{HS}=O_{P}(n^{-1/2})$, where $|||\cdot|||_{HS}$ is the Hilbert-Schmidt norm. \\
\item[(c)] $||\widehat{\beta}_{rc} - \beta|| = O_{P}(n^{-1/2})$. 
\end{itemize}
\end{theorem}

\begin{remark} \label{rem1}
Part (b) of Theorem \ref{thm1} shows that by utilizing the grid subsampling technique, we have been able to achieve parametric rates of convergence for the estimator of the covariance operator \emph{and} its generalised inverse. Part (a) shows that we have achieved consistency of the estimator of the rank even when the grid size is very large compared to $n$. This is unlike what is obtained in \cite{DP16}. Under the condition $L^{-2} = O(n^{-1})$, which is also required to ensure parametric rate of convergence of their estimator of the true covariance, it is unknown whether their estimator of rank would even be consistent. This is because the consistency of the estimator of rank in \cite{DP16} is ensured provided the exactly opposite condition holds, namely, $L^{-2}/n^{-1} \rightarrow \infty$, which implies that $L^{2}$ does not grow any faster than $n$.
\end{remark} 
\begin{proof}[Proof of Theorem \ref{thm1}]
(a) It follows from Theorem 3 in \cite{DP16} that $|\widehat{r}_{L_{*}} - r| = O_{P}((n\tau_{n})^{-1})$. Since $n\tau_{n} \rightarrow \infty$, this implies that $\widehat{r}_{L_{*}}$ converges in probability to $r$ as $n \rightarrow \infty$. Hence, for each $\epsilon > 0$, we have $P(|\widehat{r}_{L_{*}} - r| \leq \epsilon) \rightarrow 1$ as $n \rightarrow \infty$. Since both $\widehat{r}_{L_{*}}$ and $r$ takes integer values, part (a) of this theorem follows upon choosing any $\epsilon < 1$.\\
(b) Fix any $M > 0$,
\begin{eqnarray}
&& P\left(n^{1/2}|||\widehat{\mathscr{K}}_{X}^{-} - \mathscr{K}_{X}^{-}|||_{HS} > M\right) \nonumber \\
&=& P\left(n^{1/2}|||\widehat{\mathscr{K}}_{X}^{-} - \mathscr{K}_{X}^{-}|||_{HS} > M, \widehat{r}_{L_{*}} = r\right) + P\left(n^{1/2}|||\widehat{\mathscr{K}}_{X}^{-} - \mathscr{K}_{X}^{-}|||_{HS} > M, \widehat{r}_{L_{*}} \neq r\right) \nonumber \\
&\leq& P\left(n^{1/2}|||\widetilde{\mathscr{K}}_{X}^{-} - \mathscr{K}_{X}^{-}|||_{HS} > M, \widehat{r}_{L_{*}} = r\right) + P\left(\widehat{r}_{L_{*}} \neq r\right) \nonumber \\
&\leq& P\left(n^{1/2}|||\widetilde{\mathscr{K}}_{X}^{-} - \mathscr{K}_{X}^{-}|||_{HS} > M\right) + P\left(\widehat{r}_{L_{*}} \neq r\right), \label{eq1}
\end{eqnarray}
where $\widetilde{\mathscr{K}}_{X}$ is the estimator obtained in Steps 2 and 3 of the previous algorithm when the chosen rank is equal to the true rank $r$. Denote the matrix form of the minimizer in this situation by $\widetilde{K}_{X}$, i.e., 
$$ \widetilde{K}_{X} = \mbox{arg}\min_{\Theta : \mbox{rank}(\Theta) = r} ||P_{L} \circ (\widehat{K}_{W} - \Theta)|||_{F}^{2}.$$
Let us denote the eigenvalues and the eigenfunctions of $\widetilde{\mathscr{K}}_{X}$ by $\widetilde{\lambda}_{j}$'s and $\widetilde{\eta}_{j}$'s for $j=1,2,\ldots,r$. \\
\indent In order to show that $|||\widehat{\mathscr{K}}_{X}^{-} - \mathscr{K}_{X}^{-}|||_{HS} = O_{P}(n^{-1/2})$, it is enough to show that the first probability in the right hand side of \eqref{eq1} is small if we choose a large enough $M$, i.e., $|||\widetilde{\mathscr{K}}_{X}^{-} - \mathscr{K}_{X}^{-}|||_{HS} = O_{P}(n^{-1/2})$. This is because the second probability in \eqref{eq1}, namely, $P\left(\widehat{r}_{L_{*}} \neq r\right)$ converges to zero by part (a) of the theorem. Now, observe that
\begin{eqnarray}
|||\widetilde{\mathscr{K}}_{X}^{-} - \mathscr{K}_{X}^{-}|||_{HS}&\leq& \sum_{j=1}^{r} |||\widetilde{\lambda}_{j}^{-1} (\widetilde{\eta}_{j} \otimes \widetilde{\eta}_{j}) - \lambda_{j}^{-1} (\eta_{j} \otimes \eta_{j})|||_{HS} \nonumber \\
&\leq& \sum_{j=1}^{r} |\widetilde{\lambda}_{j}^{-1} - \lambda_{j}^{-1}| + \sum_{j=1}^{r} \lambda_{j}^{-1} |||(\widetilde{\eta}_{j} \otimes \widetilde{\eta}_{j}) - (\eta_{j} \otimes \eta_{j})|||_{HS} \nonumber \\
&\leq& \sum_{j=1}^{r} \widetilde{\lambda}_{j}^{-1}\lambda_{j}^{-1} |\widetilde{\lambda}_{j} - \lambda_{j}| + 2\sum_{j=1}^{r} \lambda_{j}^{-1} ||\widetilde{\eta}_{j} - \eta_{j}|| \nonumber \\
&\leq& \widetilde{\lambda}_{r}^{-1}\lambda_{r}^{-1} \max_{1 \leq j \leq r} |\widetilde{\lambda}_{j} - \lambda_{j}| + 2\lambda_{r}^{-1} \max_{1 \leq j \leq r} ||\widetilde{\eta}_{j} - \eta_{j}|| \nonumber \\
&\leq& \widetilde{\lambda}_{r}^{-1}\lambda_{r}^{-1} |||\widetilde{\mathscr{K}}_{X} - \mathscr{K}_{X}|||_{HS} + 4\sqrt{2}\lambda_{r}^{-1}a_{r}^{-1} |||\widetilde{\mathscr{K}}_{X} - \mathscr{K}_{X}|||_{HS}, \label{eq2}
\end{eqnarray}
where the last inequality follows using standard results in perturbation theory of operators (see, e.g., \cite{HE15}) and $a_{r} = \min\{\lambda_{r}, (\lambda_{1} - \lambda_{2}), \ldots, (\lambda_{r-1} - \lambda_{r})\}$. So, if we can show that $|||\widetilde{\mathscr{K}}_{X} - \mathscr{K}_{X}|||_{HS} = O_{P}(n^{-1/2})$, then using the fact that $\max_{1 \leq j \leq r} |\widehat{\lambda}_{j} - \lambda_{j}| \leq |||\widetilde{\mathscr{K}}_{X} - \mathscr{K}_{X}|||_{HS} = O_{P}(n^{-1/2})$ along with \eqref{eq2}, it will follow that $|||\widetilde{\mathscr{K}}_{X}^{-} - \mathscr{K}_{X}^{-}|||_{HS} = O_{P}(n^{-1/2})$. Consequently, from \eqref{eq1}, we will get that $|||\widehat{\mathscr{K}}_{X}^{-} - \mathscr{K}_{X}^{-}|||_{HS} = O_{P}(n^{-1/2})$.   \\
\indent We will proceed along similar lines as the proof of Theorem 3 in \cite{DP16}. Then, it follows that we only need to show that $L^{-2}||\widetilde{K}_{X} - K_{X}||_{F}^{2} = O_{P}(n^{-1})$. Define the functionals
\begin{eqnarray}
\mathbb{S}_{n,L}(\Theta) = L^{-2}||P_{L} \circ (\widehat{K}_{W} - \Theta)||_{F}^{2} \ \ \ \mbox{and} \ \ \ S_{n,L}(\Theta) = L^{-2}||P_{L} \circ (K_{W} - \Theta)||_{F}^{2},  \label{eq3}
\end{eqnarray}
where $K_{W}$ is the $L \times L$ covariance matrix of ${\bf W}_{1,L}$. Also, define $d_{n}(\Theta_{1},\Theta_{2}) = L^{-1}||\Theta_{1} - \Theta_{2}||_{F}$. This distance as well as the functionals in \eqref{eq3} are defined over the space of $L \times L$ matrices of rank equal to $r$. \\
\indent First, observe that it follows from Proposition 2 in \cite{DP16} that $K_{X}$ is the unique minimizer of $S_{n,L}(\cdot)$ almost surely over the grid provided that $L \geq 4r + 4$. Also, this minimal value is zero. Let $\gamma > 0$ and consider the following Taylor expansion:
\begin{eqnarray}
\Delta(\Theta) = S_{n,L}(\Theta)- S_{n,L}(K_{X}) = \langle S_{n,L}'(K_{X}),(\Theta - K_{X})\rangle_{F} + \frac{1}{2}\langle S_{n,L}''(\Theta_{*})(\Theta - K_{X}),(\Theta - K_{X})\rangle_{F}, \label{eq4}
\end{eqnarray}
where $\Theta_{*} = \alpha\Theta + (1-\alpha)K_{X}$ for some $\alpha \in [0,1]$. Note that $S_{n,L}'(\widecheck{\Theta}) = -2L^{-2}P_{L} \circ (K_{W} - \widecheck{\Theta})$ and $S_{n,L}''(\widetilde{\Theta})\widecheck{\Theta} = -2L^{-2}P_{L} \circ \widecheck{\Theta}$, where $\widetilde{\Theta}$ and $\widecheck{\Theta}$ are arbitrary $L \times L$ matrices of rank $r$. Observe that $S_{n,L}'(K_{X}) = 0$. So, \eqref{eq4} yields
\begin{eqnarray}
|\Delta(\Theta)| &=& L^{-2}\langle P_{L} \circ (\Theta - K_{X}),(\Theta - K_{X})\rangle_{F} \nonumber \\
&\leq& L^{-2} ||P_{L} \circ (\Theta - K_{X})||_{F}||\Theta - K_{X}||_{F} \ \leq \ L^{-2} ||\Theta - K_{X}||_{F}^{2}, \nonumber \\
\Longrightarrow && \sup_{\Theta : \mbox{rank}(\Theta) = r, d_{n}(\Theta,K_{X}) < \gamma} |\Delta(\Theta)| \ \leq \ \gamma^{2}. \label{eq5} 
\end{eqnarray}
Next, define $D(\Theta) = \mathbb{S}_{n,L}(\Theta)- S_{n,L}(\Theta) - \mathbb{S}_{n,L}(K_{X}) + S_{n,L}(K_{X})$. A first order Taylor expansion yields the following simplification of $D(\Theta)$ with $\Theta_{**} = \beta\Theta + (1-\beta)K_{X}$ for some $\beta \in [0,1]$. 
\begin{eqnarray}
|D(\Theta)| &=& |\langle \mathbb{S}_{n,L}'(\Theta_{**}),(\Theta - K_{X})\rangle_{F} - \langle S_{n,L}'(\Theta_{**}),(\Theta - K_{X})\rangle_{F}| \nonumber \\
&=& 2L^{-2}|\langle P_{L} \circ (\widehat{K}_{W} - \Theta_{**}), (\Theta - K_{X})\rangle_{F} - \langle P_{L} \circ (K_{W} - \Theta_{**}), (\Theta - K_{X})\rangle_{F}| \nonumber \\
&=& 2L^{-2}|\langle P_{L} \circ (\widehat{K}_{W} - K_{W}), (\Theta - K_{X})\rangle_{F}| \nonumber \\
&\leq& 2L^{-2}||\widehat{K}_{W} - K_{W}||_{F}||\Theta - K_{X}||_{F}. \nonumber \\
\Longrightarrow && \sup_{\Theta : \mbox{rank}(\Theta) = r, d_{n}(\Theta,K_{X}) < \gamma} |D(\Theta)| \ \leq \  2{\gamma}L^{-1}||\widehat{K}_{W} - K_{W}||_{F}.   \label{eq6}
\end{eqnarray}
Next, it can be shown that $E(L^{-2}||\widehat{K}_{W} - K_{W}||_{F}^{2}) \leq Cn^{-1}$ for a constant $C = \sup_{s,t \in [0,1]^{2}} \mathrm{var}[W(s)W(t)]$, and the finiteness of this constant is a consequence of the assumption $E(||W||^{4}) < \infty$. Thus,
\begin{eqnarray}
E\left\{\sup_{\Theta : \mbox{rank}(\Theta) = r, d_{n}(\Theta,K_{X}) < \gamma} |D(\Theta)|\right\} \leq 2{\gamma}\sqrt{C/n}. \label{eq7}
\end{eqnarray}
It now follows from Theorem 3.4.1 in \cite{VW96} that the minimizer $\widetilde{K}_{X}$ of $\mathbb{S}_{n,L}(\Theta)$  satisfies 
\begin{eqnarray}
nd_{n}^{2}(\widetilde{K}_{X},K_{X}) &=& O_{P}(1)  \ \ \Longrightarrow \ \  L^{-1}||\widetilde{K}_{X} - K_{X}||_{F} = O_{P}(n^{-1/2}) \label{eq8}
\end{eqnarray}
as $n \rightarrow \infty$, where the $O_{P}(1)$ term is uniform in $L$. This and the fact that $L^{-2} = O(n^{-1})$  completes the proof of the fact that $|||\widetilde{\mathscr{K}}_{X} - \mathscr{K}_{X}|||_{HS} = O_{P}(n^{-1/2})$. Hence, from the arguments given towards the beginning, if follows that both of $|||\widehat{\mathscr{K}}_{X}^{-} - \mathscr{K}_{X}^{-}|||_{HS}$ and $|||\widehat{\mathscr{K}}_{X} - \mathscr{K}_{X}|||_{HS}$ are $O_{P}(n^{-1/2})$ as $n \rightarrow \infty$.  \\
\indent For proving the second part of (b) of this theorem, note that $\widehat{C}_{y,W} = n^{-1} \sum_{i=1}^{n} y_{i}W_{i} - \bar{y}\bar{W}$. It now follows from the central limit theorem in separable Hilbert spaces (see, e.g., \cite{Bosq00}) that $||\widehat{C}_{y,W} - \mathrm{cov}(y,W)|| = O_{P}(n^{-1/2})$ as $n \rightarrow \infty$ from the weak law of large numbers in a separable Hilbert space (see, e.g., \cite{Bosq00}). Now, recalling that $\beta = \mathscr{K}_{X}^{-}\mathrm{cov}(y,W)$ from the discussion in the Introduction, and using the earlier statement and the first part of (b) of this theorem, we get
\begin{eqnarray*}
||\widehat{\beta}_{rc} - \beta|| &\leq& |||\widehat{\mathscr{K}}_{X}^{-} - \mathscr{K}_{X}^{-}|||_{HS}||\widehat{C}_{y,W}|| + |||\mathscr{K}_{X}^{-}|||_{HS}||\widehat{C}_{y,W} - \mathrm{cov}(y,W)|| \ = \ O_{P}(n^{-1/2})
\end{eqnarray*}
as $n \rightarrow \infty$. This completes the proof of the theorem.
\end{proof}

\subsection{Function-on-function model} \label{sec2-2}
\indent We now consider the case when the response variable $y$ is also functional. In this case, the true function-on-function linear regression is given by $y = \alpha + \mathscr{B}X + \epsilon$, where $\epsilon$ is now a random element in the underlying separable Hilbert space, and $\mathscr{B}$ is the unknown (bounded and linear) slope operator. Otherwise, we still only get to observe the corrupted proxy $W=X+U$ instead of $X$ itself, with $X$ and $U$ satisfying the same assumptions as before. If $\mathscr{K}_X$ were known, the least squares estimator of $\mathscr{B}$ would be given by the solution of the equation $\mathscr{B}\mathscr{K}_{X} = \mathrm{cov}(y,X)$, which equals $\mathrm{cov}(y,X)\mathscr{K}_{X}^{-}$ with $\mathrm{cov}(y,X) = E(y \otimes X)$. This is due to the identifiability constraint that $\mathscr{B} \in span\{(\eta_{i} \otimes \eta_{j}) : 1 \leq i,j \leq r\}$ so that $\mathscr{B}\mathscr{P}_{X} = \mathscr{P}_{X}\mathscr{B} = \mathscr{B}$, where $\mathscr{P}_{X}$ is the projection operator onto $span\{\eta_{j} : 1\leq j \leq r\}$. Also note that $\mathrm{cov}(y,X) = \mathrm{cov}(y,W)$.\\

\noindent This motivates the regression calibration estimator in the functional response case, defined as 
$$\mathscr{B}_{rc} = \widehat{\mathscr{C}}_{y,W}\widehat{\mathscr{K}}_{X}^{-},$$
where $\widehat{\mathscr{C}}_{y,W} = n^{-1} \sum_{i=1}^{n} y_{i} \otimes W_{i} - \bar{y} \otimes \bar{W}$ is the empirical covariance operator between the $y_{i}$'s and the $W_{i}$'s, and the estimator $\widehat{\mathscr{K}}_{X}$ is exactly as in the scalar response case considered earlier. We then have:
\begin{theorem} \label{thm2}
Under the assumptions of Theorem \ref{thm1} with $E(\epsilon^{2}) < \infty$ replaced by $E(||\epsilon||^{2}) < \infty$, we have $|||\widehat{\mathscr{B}}_{rc} - \mathscr{B}|||_{HS} = O_{P}(n^{-1/2})$ as $n \rightarrow \infty$.
\end{theorem}
\begin{proof}[Proof of Theorem \ref{thm2}]
Note that $|||\mathrm{cov}(y,W)|||_{HS}^{2} \leq E(||y||^{2})E(||W||^{2}) < \infty$. By the central limit theorem for separable Hilbert spaces (see, e.g., \cite{Bosq00}) applied to the space of Hilbert-Schmidt operators, it follows that $|||\widehat{\mathscr{C}}_{y,W} - \mathrm{cov}(y,W)|||_{HS} = O_{P}(n^{-1/2})$ as $n \rightarrow \infty$. So,
\begin{eqnarray*}
&&|||\widehat{\mathscr{B}}_{rc} - \mathscr{B}|||_{HS} \\
&\leq& |||(\widehat{\mathscr{C}}_{y,W} - \mathrm{cov}(y,W))(\widehat{\mathscr{K}}_{X}^{-} - \mathscr{K}_{X}^{-})|||_{HS} + |||(\widehat{\mathscr{C}}_{y,W} - \mathrm{cov}(y,W))\mathscr{K}_{X}^{-})|||_{HS} \\
&& \ + \ |||\mathrm{cov}(y,W)(\widehat{\mathscr{K}}_{X}^{-} - \mathscr{K}_{X}^{-})|||_{HS} \\
&\leq& |||\widehat{\mathscr{C}}_{y,W} - \mathrm{cov}(y,W)|||_{HS}|||\widehat{\mathscr{K}}_{X}^{-} - \mathscr{K}_{X}^{-}|||_{HS} + |||\widehat{\mathscr{C}}_{y,W} - \mathrm{cov}(y,W)|||_{HS}|||\mathscr{K}_{X}^{-}|||_{HS} \\
&& \ + \ |||\mathrm{cov}(y,W)|||_{HS}|||\widehat{\mathscr{K}}_{X}^{-} - \mathscr{K}_{X}^{-}|||_{HS}.
\end{eqnarray*}
The right hand side of the last inequality above is $O_{P}(n^{-1/2})$ by the earlier fact and part (b) of Theorem \ref{thm1}. This completes the proof of the theorem.
\end{proof}

\subsection{Practical implementation} \label{subsec2-1}

\indent To implement the grid subsampling technique, it is useful in practice to find several $\widehat{r}_{L_{*}}$'s in Step 1 of earlier algorithm, each corresponding to a different random selection of subgrid points $s_{1}, s_{2}, \ldots, s_{L_{*}}$ from the original grid $t_{1}, t_{2}, \ldots, t_{L}$. This helps reduce sampling bias and thus results in a more stable estimator. Suppose we find $B$ such separate estimators, with the estimators at the $b$th iteration being denoted by $\widehat{r}_{b,L_{*}}$. The final estimator is given by the mode of the empirical distribution of the $\widehat{r}_{b,L_{*}}$'s, namely, 
$$\widetilde{r}_{L_{*}} := \mbox{arg}\max_{1 \leq q \leq (L_{*}/4 - 1)} B^{-1} \sum_{b=1}^{B} \mathbf{1}(\widehat{r}_{b,L_{*}} = q).$$
Note that the upper bound on the range of $q$ values is enforced by the identifiability condition discussed earlier.
\begin{proposition} \label{prop1}
For any $B \geq 1$,  $\widetilde{r}_{L_{*}}$ converges in probability to $r$ as $n \rightarrow \infty$.
\end{proposition}
\begin{remark} \label{rem2}
The above proposition shows that even if we use $\widetilde{r}_{L_{*}}$ instead of $\widehat{r}_{L_{*}}$ in the definition of $\widehat{\mathscr{K}}_{X}$, the resulting estimator of $\mathscr{K}_{X}$, and consequently that of $\beta$ (or $\mathscr{B}$ in the function-on-function setting), will be $\sqrt{n}$-consistent. To see this, note that the proof of Theorems \ref{thm1} and \ref{thm2} only utilizes the fact that $\widehat{r}_{L_{*}}$ is consistent for $r$ and that $L_{*}$ is fixed. Thus, the proof goes through unchanged for $\widetilde{r}_{L_{*}}$ by Proposition \ref{prop1}.
\end{remark}
\begin{proof}[Proof of Proposition \ref{prop1}]
Fix $q \neq r$. Since for each $1 \leq b \leq B$, $\widehat{r}_{b,L_{*}}$ converges to $r$ by part (a) of Theorem \ref{thm1}, it follows that $B^{-1} \sum_{b=1}^{B} \mathbf{1}(\widehat{r}_{b,L_{*}} = q) \rightarrow 0$ as $n \rightarrow \infty$. Also, $B^{-1} \sum_{b=1}^{B} \mathbf{1}(\widehat{r}_{b,L_{*}} = r) \rightarrow 1$ as $n \rightarrow \infty$. Thus, for any $q \neq r$, we have $\mathbf{1}\left\{\sum_{b=1}^{B} \mathbf{1}(\widehat{r}_{b,L_{*}} = r) > \sum_{b=1}^{B} \mathbf{1}(\widehat{r}_{b,L_{*}} = q)\right\} \rightarrow 1$ in probability as $n \rightarrow \infty$. So, $\mathbf{1}\left\{\sum_{b=1}^{B} \mathbf{1}(\widehat{r}_{b,L_{*}} = r) > \max_{1 \leq q \leq (L_{*}/4 - 1), q \neq r} \sum_{b=1}^{B} \mathbf{1}(\widehat{r}_{b,L_{*}} = q)\right\} \rightarrow 1$ in probability as $n \rightarrow \infty$. But this is the same as saying $I(\widetilde{r}_{L_{*}} = r) \rightarrow 1$, i.e., $\widetilde{r}_{L_{*}}$ converges to $r$ in probability as $n \rightarrow \infty$.
\end{proof}
\indent This being said, the algorithm for constructing the grid sub-sampling estimator of $\mathscr{K}_{X}$ in practice is thus as follows:
\begin{itemize}
\item[Step 1*:] For the $b$th iteration ($1 \leq b \leq B$), randomly select a sub-grid $s_{1}, s_{2}, \ldots, s_{L_{*}}$ of size $L_{*}$ from the original grid $t_{1}, t_{2}, \ldots, t_{L}$. Assuming that the original grid is itself adequate, one may do this as follows. Set $L_*= \lfloor L/m \rfloor$ for some pre-chosen fixed integer $m>1$ such that $\lfloor L/m \rfloor>4(r+1)$ (which is always possible provided the original grid is sufficiently large, because $r$ is finite). Then define $I_{j,L_{*}} =  \bigcup_{p=m(j-1)+1}^{mj} I_{p,L}$ for $1 \leq j \leq L_{*}$, and select the $j$th sub-grid node $s_{j}$ uniformly at random among the values $t_{m(j-1)+1},\ldots,t_{mj}$. \\
\indent Compute the empirical covariance $\widehat{K}_{W*}$ of $W$ for this chosen grid.

\item[Step 2*:] Find the value of $f_{L_{*}}(j) = \min |||P_{L_{*},\delta_{*}} \circ (\widehat{K}_{W*} - \Theta)|||_{F}^{2}$ over all $L_{*} \times L_{*}$ positive definite matrices $\Theta$ of rank $j$, for each $j = 1,2,\ldots,M$, where $1 \leq M \leq (L_{*}/4) - 1$ is a pre-chosen integer independent of the sub-grid. Here, the $(i,j)$th element of $P_{L_{*},\delta_{*}}$ equals $\mathbf{1}(|i - j| > \lceil L_{*}\delta_{*}\rceil)$ for a pre-chosen $\delta_{*} \in [0,1/4]$. The parameter $\delta_{*}$ is an upper bound on the bandwidth allowed in practice. The minimisation of $f_{L_{*}}(j)$ over $\Theta$ can be carried out by a quasi-Newton method (e.g. using the function \texttt{fminunc} in \texttt{MATLAB} or the function \texttt{optim} in \texttt{R}), with starting value equal to the rank $j$ projection of $\hat{K}_W$ obtained via SVD. 

\item[Step 3*:] Next, set $\widehat{r}_{b,L_{*}} = \min\{j : f_{L_{*}}(j) \leq c_{1}\}$, where $c_{1} > 0$ is a pre-chosen cut-off level for the scree plot $j \mapsto f_{L_{*}}(j)$. Note that $c_n$ is in 1-1 correspondence with $\tau_n$, and the resulting $\widehat{r}_{b,L_{*}}$ is the same as what would be obtained by solving \eqref{DPestimator} instead (see \cite{DP16} for a rigorous proof).

\item[Step 4*:] Compute the mode of the empirical distribution of the $\widehat{r}_{b,L_{*}}$'s, namely, 
$$\widetilde{r}_{L_{*}} = \mbox{arg}\max_{1 \leq q \leq (L_{*}/4 - 1)} B^{-1} \sum_{b=1}^{B} \mathbf{1}(\widehat{r}_{b,L_{*}} = q).$$

\item[Step 5*:] Compute the empirical covariance $\widehat{K}_{W}$ of $W$ for the full grid of size $L$, and set $\widehat{K}_{X} = \mbox{arg}\min |||P_{L} \circ (\widehat{K}_{W} - \Theta)|||_{F}^{2}$ over all $L \times L$ positive definite matrices $\Theta$ of rank $\widetilde{r}_{L_{*}}$ (again, by a quasi-Newton method). Construct the estimator  $\widehat{k}_{X}$ of the kernel $k_X$ of $\mathscr{K}_X$ as $ \widehat{k}_{X}(s,t) = \sum_{i=1}^{L}\sum_{j=1}^{L} \widehat{K}_{X}(i,j)\mathbf{1}\{(s,t) \in I_{i,L} \times I_{j,L}\}$, and the estimators $ \widehat{\mathscr{K}}_{X} = \sum_{j=1}^{\widehat{r}_{L_{*}}} \widehat{\lambda}_{j} (\widehat{\eta}_{j} \otimes \widehat{\eta}_{j})$ and $ \widehat{\mathscr{K}}_{X}^{-} = \sum_{j=1}^{\widehat{r}_{L_{*}}} \widehat{\lambda}_{j}^{-1} (\widehat{\eta}_{j} \otimes \widehat{\eta}_{j})$ where $\widehat{\lambda}_{j}$ and $\widehat{\eta}_{j}$ denote the eigenvalues and the eigenfunctions of the kernel $\widehat{k}_{X}(s,t)$.
\end{itemize}


\section{Regression calibration method for functional quadratic regression}   \label{fqr}

We now demonstrate how one can extend our regression calibration estimator to the case of functional quadratic regression with functional measurement error. \cite{YM10} studied this model for functional data and demonstrated its utility using the well-known Tecator dataset that contains spectrometry measurements (see, e.g., \cite{RS05}). The model for the functional quadratic regression with scalar response and true covariate $X \in L_{2}[0,1]$ is given by 
$$ y = \alpha + \int_{0}^{1} X(t)\beta(t)dt + \int_{0}^{1} \int_{0}^{1} b(t,s)X(t)X(s)dtds.$$
Consider an operator $\mathscr{B}$ on $L_{2}[0,1]$ defined as $\mathscr{B}f(t) = \int_{0}^{1} b(s,t)f(s)ds$ for $f \in L_{2}[0,1]$. Assuming that $b(\cdot,\cdot) \in L_{2}([0,1]^{2})$, it follows that $\mathscr{B}$ is a Hilbert-Schmidt operator. Recall that the inner product between two Hilbert-Schmidt operators $\mathscr{T}$ and $\mathscr{S}$ is given by $\langle \mathscr{T},\mathscr{S}\rangle_{HS} = \mbox{tr}(\mathscr{S}^{*}\mathscr{T})$. Using this definition, it follows that 
\begin{eqnarray*}
\int_{0}^{1} \int_{0}^{1} b(t,s)X(t)X(s)dtds = \langle \mathscr{B}X, X\rangle = \langle X, \mathscr{B}^{*}X\rangle = \mbox{tr}(X \otimes \mathscr{B}^{*}X) = \mbox{tr}(\mathscr{B}^{*}(X \otimes X)) = \langle X \otimes X, \mathscr{B} \rangle_{HS}.
\end{eqnarray*}
Thus, we obtain the following alternative representation of the regression model using operator notation.
$$ y = \alpha + \langle X,\beta\rangle + \langle X \otimes X, \mathscr{B} \rangle_{HS} + \epsilon.$$
This formulation also shows how quadratic regression is equivalent to linear regression with covariates and corresponding slope parameters lying in different Hilbert spaces. As in the previous section, we will assume for identifiability that $\beta$ and $\mathscr{B}$ lie in the range spaces of $\mathscr{K}_{X}$ and $\mathrm{var}(X \otimes X)$, respectively. We assume that the true covariate $X$ is Gaussian. Define $\zeta_{jj'} = \eta_{j} \otimes \eta_{j'}$ for $1 \leq j,j' \leq r$. Denote the tensor product between two Hilbert-Schmidt operators by $\otimes_{2}$ and the inner product on this tensor space by $\langle \cdot,\cdot\rangle_{2}$ and the norm by $|||\cdot|||_{2}$. Then, direct calculations show that 
\begin{eqnarray*}
\mathrm{var}(X \otimes X) &=& E\{(X \otimes X) \otimes_{2} (X \otimes X)\} - \mathscr{K}_{X} \otimes_{2} \mathscr{K}_{X} \\
&=& 2\sum_{j=1}^{r} \lambda_{j}^{2} (\zeta_{jj} \otimes_{2} \zeta_{jj}) + \sum_{1 \leq j < j' \leq r} \lambda_{j}\lambda_{j'} \{\zeta_{jj'} \otimes_{2} \zeta_{jj'} + \zeta_{jj'} \otimes_{2} \zeta_{j'j} + \zeta_{j'j} \otimes_{2} \zeta_{j'j} + \zeta_{j'j} \otimes \zeta_{jj'}\}.
\end{eqnarray*}
Observe that for indices $(i_{1}, i_{2}, i_{3}, i_{4})$ and $(j_{1}, j_{2}, j_{3}, j_{4})$, we have 
$$ \langle \zeta_{i_{1}i_{2}} \otimes_{2} \zeta_{i_{3}i_{4}}, \zeta_{j_{1}j_{2}} \otimes \zeta_{j_{3}j_{4}}\rangle_{2} = \langle \zeta_{i_{1}i_{2}},\zeta_{j_{1}j_{2}}\rangle_{HS}\langle \zeta_{i_{3}i_{4}},\zeta_{j_{3}j_{4}}\rangle_{HS} =  \delta_{i_{1}j_{1}}\delta_{i_{2}j_{2}}\delta_{i_{3}j_{3}}\delta_{i_{4}j_{4}},$$
where $\delta_{ij}$ is the Kronecker delta symbol. Thus, $\{\zeta_{j_{1}j_{2}} \otimes_{2} \zeta_{j_{3}j_{4}} : 1 \leq j_{1}, j_{2}, j_{3}, j_{4} \leq r\}$ forms an orthonormal system in this tensor product space. Further, $\mathrm{cov}(\langle X \otimes X, \mathscr{B}\rangle_{HS}, X) = 0$ and $\mathrm{cov}(\langle X,\beta\rangle, X \otimes X) = 0$ by symmetry arguments. Note that these are zero elements in different spaces. The formulation as a linear regression problem and with the above facts yields the following least square functional normal equation. 
\[
\begin{bmatrix}
\mathscr{K}_{X} & 0 \\
0 & \mathrm{var}(X \otimes X) 
\end{bmatrix}
\begin{bmatrix}
\beta \\
\mathscr{B}
\end{bmatrix}
=
\begin{bmatrix}
\mathrm{cov}(y,X) \\
\mathrm{cov}(y,X \otimes X)
\end{bmatrix}.
\]
Thus, the least square solution for $\beta$ and $\mathscr{B}$ are given by
\[
\begin{bmatrix}
\widetilde{\beta} \\
\widetilde{\mathscr{B}}
\end{bmatrix}
=
\begin{bmatrix}
\mathscr{K}_{X}^{-} & 0 \\
0 & (\mathrm{var}(X \otimes X))^{-}
\end{bmatrix}
\begin{bmatrix}
\mathrm{cov}(y,X) \\
\mathrm{cov}(y,X \otimes X)
\end{bmatrix}.
\]
\indent However, we do not observe the true covariate $X$, but observe $W = X + U$ as considered in the previous section. Note that we do not impose any distributional assumption on $U$ (except bandedness of its covariance) so that the distribution of $W$ is arbitrary. We use the method discussed in the previous section to find $\widehat{\mathscr{K}}_{X}$ and consequently, its eigenvalues and eigenfunctions that are denoted by $\widehat{\lambda}_{j}$ and $\widehat{\eta}_{j}$. Define $\widehat{\zeta}_{jj'} = \widehat{\eta}_{j} \otimes \widehat{\eta}_{j'}$ for $1 \leq j,j' \leq \widetilde{r}_{L_{*}}$, where $\widetilde{r}_{L_{*}}$ is the rank of $\widehat{\mathscr{K}}_{X}$ as defined in the previous section. Define
$$ \widehat{\mathrm{var}(X \otimes X)} = 2\sum_{j=1}^{\widetilde{r}_{L_{*}}} \widehat{\lambda}_{j}^{2} (\widehat{\zeta}_{jj} \otimes_{2} \widehat{\zeta}_{jj}) + \sum_{1 \leq j < j' \leq \widetilde{r}_{L_{*}}} \widehat{\lambda}_{j}\widehat{\lambda}_{j'} \{\widehat{\zeta}_{jj'} \otimes_{2} \widehat{\zeta}_{jj'} + \widehat{\zeta}_{jj'} \otimes_{2} \widehat{\zeta}_{j'j} + \widehat{\zeta}_{j'j} \otimes_{2} \widehat{\zeta}_{j'j} + \widehat{\zeta}_{j'j} \otimes \widehat{\zeta}_{jj'}\}.$$
So, we get that
$$ (\widehat{\mathrm{var}(X \otimes X)})^{-} = 2\sum_{j=1}^{\widetilde{r}_{L_{*}}} \widehat{\lambda}_{j}^{-2} (\widehat{\zeta}_{jj} \otimes_{2} \widehat{\zeta}_{jj}) + \sum_{1 \leq j < j' \leq \widetilde{r}_{L_{*}}} \widehat{\lambda}_{j}^{-1}\widehat{\lambda}_{j'}^{-1} \{\widehat{\zeta}_{jj'} \otimes_{2} \widehat{\zeta}_{jj'} + \widehat{\zeta}_{jj'} \otimes_{2} \widehat{\zeta}_{j'j} + \widehat{\zeta}_{j'j} \otimes_{2} \widehat{\zeta}_{j'j} + \widehat{\zeta}_{j'j} \otimes \widehat{\zeta}_{jj'}\}.$$
We assume that $X$ and $U$ satisfy a higher order uncorrelatedness condition in the sense that $$E\{X(t_{1})^{a_{1}}X(t_{2})^{a_{2}}X(t_{3})^{a_{3}}U(t_{4})^{a_{4}}U(t_{5})^{a_{5}}\} = E\{X(t_{1})^{a_{1}}X(t_{2})^{a_{2}}X(t_{3})^{a_{3}}\}E\{U(t_{4})^{a_{4}}U(t_{5})^{a_{5}}\}$$
for all choices of $t_{j} \in [0,1]$ and $a_{j} \in \{0,1\}, 1 \leq j \leq 5$, satisfying $\sum_{j} a_{j} \leq 4$. This is obviously true if $X$ and $U$ are independent but is much weaker than assuming independence. Under this assumption, it can be shown that $\mathrm{cov}(y,X \otimes X) = \mathrm{cov}(y,W \otimes W)$. We can then estimate $\mathrm{cov}(y,X \otimes X)$ by $\widehat{\mathscr{C}}_{y,W \otimes W}$, which is the empirical covariance between the $y_{i}$'s and the $W_{i} \otimes W_{i}$'s. Thus, the regression calibration estimators of $\beta$ and $\mathscr{B}$ for a functional quadratic regression model are defined as
\[
\begin{bmatrix}
\widehat{\beta}_{rc} \\
\widehat{\mathscr{B}}_{rc}
\end{bmatrix}
=
\begin{bmatrix}
\widehat{\mathscr{K}}_{X}^{-} & 0 \\
0 & (\widehat{\mathrm{var}(X \otimes X)})^{-}
\end{bmatrix}
\begin{bmatrix}
\widehat{C}_{yW} \\
\widehat{\mathscr{C}}_{y,W \otimes W}
\end{bmatrix}.
\]
Thus, for the functional quadratic regression, the estimator of the slope parameter that takes values in the underlying Hilbert space stays the same as that in the functional linear regression. The estimator of the slope parameter in the Hilbert-Schmidt space can be calculated using the following simplified expression.
\begin{eqnarray*} 
\widehat{\mathscr{B}}_{rc} &=& (\widehat{\mathrm{var}(X \otimes X)})^{-}\widehat{\mathscr{C}}_{y,W \otimes W} \\
&=&  2\sum_{j=1}^{\widetilde{r}_{L_{*}}} \widehat{\lambda}_{j}^{-2} \langle \widehat{\zeta}_{jj},\widehat{\mathscr{C}}_{y,W \otimes W}\rangle_{HS}~\widehat{\zeta}_{jj} + \sum_{1 \leq j < j' \leq \widetilde{r}_{L_{*}}} \widehat{\lambda}_{j}^{-1}\widehat{\lambda}_{j'}^{-1} \{ \langle \widehat{\zeta}_{jj'},\widehat{\mathscr{C}}_{y,W \otimes W}\rangle_{HS}(\widehat{\zeta}_{jj'} + \widehat{\zeta}_{j'j}) \\
&& \hspace{8cm} + \ \langle \widehat{\zeta}_{j'j},\widehat{\mathscr{C}}_{y,W \otimes W}\rangle_{HS}(\widehat{\zeta}_{j'j} + \widehat{\zeta}_{jj'})\} \\
&=& 2\sum_{j=1}^{\widetilde{r}_{L_{*}}} \widehat{\lambda}_{j}^{-2} \langle\widehat{\mathscr{C}}_{y,W \otimes W}\widehat{\eta}_{j},\widehat{\eta}_{j}\rangle~\widehat{\zeta}_{jj} + \sum_{1 \leq j < j' \leq \widetilde{r}_{L_{*}}} \widehat{\lambda}_{j}^{-1}\widehat{\lambda}_{j'}^{-1} \{ \langle\widehat{\mathscr{C}}_{y,W \otimes W}\widehat{\eta}_{j},\widehat{\eta}_{j'}\rangle(\widehat{\zeta}_{jj'} + \widehat{\zeta}_{j'j}) \\
&& \hspace{8cm} + \ \langle\widehat{\mathscr{C}}_{y,W \otimes W}\widehat{\eta}_{j'},\widehat{\eta}_{j}\rangle(\widehat{\zeta}_{j'j} + \widehat{\zeta}_{jj'})\} \\
&=& 2\sum_{j=1}^{\widetilde{r}_{L_{*}}} \widehat{\lambda}_{j}^{-2} \langle\widehat{\mathscr{C}}_{y,W \otimes W}\widehat{\eta}_{j},\widehat{\eta}_{j}\rangle~\widehat{\zeta}_{jj} + \sum_{1 \leq j \neq j' \leq \widetilde{r}_{L_{*}}} \widehat{\lambda}_{j}^{-1}\widehat{\lambda}_{j'}^{-1} \langle\widehat{\mathscr{C}}_{y,W \otimes W}\widehat{\eta}_{j},\widehat{\eta}_{j'}\rangle(\widehat{\zeta}_{jj'} + \widehat{\zeta}_{j'j}).
\end{eqnarray*}
The third equality above follows from the fact that for any Hilbert-Schmidt operator $\mathscr{T}$ and elements $u,v$ in the underlying Hilbert space, we have $\langle u \otimes v,\mathscr{T}\rangle_{HS} = \mbox{tr}(\mathscr{T}^{*}(u \otimes v)) = \langle \mathscr{T}u,v\rangle$. The last equality above follows form the fact that $\widehat{\mathscr{C}}_{y,W \otimes W} = \widehat{\mathscr{C}}_{y,W \otimes W}^{*}$, i.e., it is a self-adjoint operator. \\
\indent We now study the asymptotic behaviour of the regression calibration estimators for the functional quadratic model. We will need the following notation and definitions. Let ${\cal H}$ be the vector space containing vector elements of the form $(u, \mathscr{T})$, where $u$ belongs to the underlying separable Hilbert space of functions and $\mathscr{T}$ is a Hilbert-Schmidt operator on it. We equip ${\cal H}$ with the inner product $\langle (u_{1},\mathscr{T}_{1}), (u_{2},\mathscr{T}_{2})\rangle_{D} = \langle u_{1}, u_{2}\rangle + \langle \mathscr{T}_{1},\mathscr{T}_{2}\rangle_{HS}$. The associated norm is denoted by $|||\cdot|||_{D}$. Here the subscript 'D' represents the fact that ${\cal H}$ is the direct sum of the two spaces. This inner product makes ${\cal H}$ a separable Hilbert space. Note that $(\beta,\mathscr{B})$ as well as its estimator $(\widehat{\beta},\widehat{\mathscr{B}})$ are elements of ${\cal H}$. We have the following result.
\begin{theorem}  \label{fqr-thm}
Suppose that the assumptions of Theorem \ref{thm1} hold. Then, all of $|||(\widehat{\mathrm{var}(X \otimes X)})^{-} - (\mathrm{var}(X \otimes X))^{-}|||_{2}$, $|||(\widehat{\mathrm{var}(X \otimes X)}) - (\mathrm{var}(X \otimes X))|||_{2}$ and $|||(\widehat{\beta},\widehat{\mathscr{B}})|||_{D}$ are $O_{P}(n^{-1/2})$ as $n \rightarrow \infty$.
\end{theorem}
\begin{proof}[Proof of Theorem \ref{fqr-thm}]
For proving the first part of the theorem, we can argue along similar lines as in the proof of Theorem \ref{thm1} and work with the estimator, say $\widetilde{\mathrm{var}(X \otimes X)})$, obtained when the chosen $\widetilde{r}_{L_{*}}$ equals the true rank $r$. In that case,
$$(\widetilde{\mathrm{var}(X \otimes X)})^{-} = 2\sum_{j=1}^{r} \widetilde{\lambda}_{j}^{-2} (\widetilde{\zeta}_{jj} \otimes_{2} \widetilde{\zeta}_{jj}) + \sum_{1 \leq j < j' \leq r} \widetilde{\lambda}_{j}^{-1}\widetilde{\lambda}_{j'}^{-1} \{\widetilde{\zeta}_{jj'} \otimes_{2} \widetilde{\zeta}_{jj'} + \widetilde{\zeta}_{jj'} \otimes_{2} \widetilde{\zeta}_{j'j} + \widetilde{\zeta}_{j'j} \otimes_{2} \widetilde{\zeta}_{j'j} + \widetilde{\zeta}_{j'j} \otimes \widetilde{\zeta}_{jj'}\}.$$
Since the estimators $\widehat{\lambda}_{j}$'s and $\widehat{\zeta}_{jj'}$'s are consistent for their population counterparts uniformly in $j$ and $j'$, it follows from arguments similar to those used in the proof of Theorem \ref{thm1} that both $|||(\widetilde{\mathrm{var}(X \otimes X)})^{-} - (\mathrm{var}(X \otimes X))^{-}|||_{2}$ and $|||(\widetilde{\mathrm{var}(X \otimes X)}) - (\mathrm{var}(X \otimes X))|||_{2}$ converge to zero in probability as $n \rightarrow \infty$. The first part of the theorem now follows. \\
\indent It follows from the definition of $|||\cdot|||_{D}$ that the proof of the second part of the theorem will be complete if we can show that $|||\widehat{\mathscr{B}} - \mathscr{B}|||_{HS} = O_{P}(n^{-1/2})$ as $n \rightarrow \infty$ (since we have already shown in Theorem \ref{thm2} that $||\widehat{\beta} - \beta|| = O_{P}(n^{-1/2})$ as $n \rightarrow \infty$). This follows using arguments similar to those used in the proof of Theorem \ref{thm2} and using the first part of the present theorem.
\end{proof}

\section{Numerical Experiments} \label{sec3}

In this section, we consider some simulated models with scalar response to assess the performance of the regression calibration estimator $\widehat{\beta}_{rc}$. We compare it with the spectral truncation estimator considered by \cite{HH07}, when the latter is applied in the setting of the mis-specified model $y = a + \langle W,\beta\rangle + e$, i.e., ignoring the presence of measurement error in the data. The random elements $Y$ and $Z$ are expressed as $X = \sum_{j=1}^{r} \lambda_{j}^{1/2} P_{j} \eta_{j}$ and $U = \sum_{l=1}^{D} \gamma_{l}^{1/2} Q_{l} \phi_{l}$. We have considered three models for $Y$ as follows: \\
(M1) $r = 3$; $\eta_{1} \equiv 1, \eta_{2}(t) = \sqrt{2}\sin(2{\pi}t)$, and $\eta_{3}(t) = \sqrt{2}\cos(2{\pi}t)$  for $t \in [0,1]$. \\
(M2) $r = 5$; $\eta_{j}(t)$'s are the first $r$ normalized functions in the Gram-Schmidt orthogonalization of the functions $f_{1}(t) = 5t\sin(2{\pi}t), f_{2}(t) = t\cos(2{\pi}t) - 3, f_{3}(t) = 5t + \sin(2{\pi}t) - 2, f_{4}(t) = \cos(4{\pi}t) + 0.25t^{2}$, and $f_{5}(t) = 6t(1-t)$ for $t \in [0,1]$. \\
(M3) $r = 5$; $\eta_{j}(t)$'s are the normalized versions of the first $r$ shifted Legendre polynomials, namely, $f_{1}(t) \equiv 1, f_{2}(t) = 2t - 1, f_{3}(t) = 6t^{2} - 6t + 1, f_{4}(t) = 20t^{3} - 30t^{2} + 12t - 1$, and $f_{5}(t) = 70t^{4} - 140t^{3} + 90t^{2} - 20t + 1$ for $t \in [0,1]$. \\
In all cases, the $\lambda_{j}$'s are decreasing and chosen to be equispaced between $1.5$ and $0.3$. The $P_{j}$'s and the $Q_{l}$'s are i.i.d. standard random variables. We chose two different values of the bandwidth $\delta$, namely, $\delta = 0.05$ and $\delta = 0.1$, and consequently two different models for $U$. Also, we chose $D = \lfloor(1/\delta)\rfloor$, $\gamma_{1} = 0.09$, and $\gamma_{l}$'s to be  decreasing and equispaced between $0.04$ and $0.01$ for $l=2,3,\ldots,D$. Each $\phi_{j}$ is triangular function with unit norm and and is supported on $[(j-1)\delta,j\delta]$ for $j=1,2,\ldots,D$. The magnitudes of the $\gamma_{l}$'s are chosen to be smaller than the smallest $\lambda_{j}$ so that the measurement error $U$ does not overwhelm the signal $X$, which is typically the case in practice. For generating the response variable, we considered three different choices of the slope parameter corresponding to the three models for $X$. They are $\beta = \eta_{1} + \eta_{2} - \eta_{3}$ for model (M1), $\beta = -0.4\eta_{1} + 2\eta_{2} - \eta_{3} + \eta_{4} - 0.7\eta_{5}$ for model (M2), and $\beta = 0.7\eta_{1} + 3\eta_{2} - \eta_{4} + 0.5\eta_{5}$ for model (M3). For all the models, the error component $\epsilon$ in the functional linear regression is chosen to have a standard normal distribution. \\
\indent Each functional random element in observed over a grid of $100$ points in $[0,1]$ which are chosen by uniformly sampling a single point from each of the intervals $[0.01(j-1),0.01j], j=1,2,\ldots,100$. This grid is then kept fixed throughout our simulations to reduce computational cost. The sample size is $n = 100$. For selecting the number of principal components to be retained in the spectral truncation estimator, we employed a $2$-fold cross-validation technique to minimize the prediction error. The number of components are allowed to vary over the range $1$ to $10$ and for each value, the prediction error is averaged over $500$ sample partitions. For computing the regression calibration estimator, we chose $L_{*} = 25$ and computed the estimator $\widetilde{r}_{L_{*}}$ described in Section \ref{sec2} with $B = 100$. Then Step $2^*$ in Section \ref{subsec2-1} is implemented by selecting $M = 10$ and $c_{1} = 0.01L_{*}^{2}$. The latter implies that the standardized Hilbert-Schmidt norm of the error is allowed to be at most $0.01$.  Also, we chose $\delta_{*} = 0.15$ in Step $2^*$ for all our simulations. This choice already implies a significant dependence of the error across the domain of the functional data and mimics practical situations quite well. 
In all the models, the true rank was correctly estimated. Further, the fact that the rank has been correctly estimated for both choices of $\delta$ suggests a type of robustness of the estimator to mis-specification of the upper bound $\delta_{*}$, which is more so when the true $\delta = 0.05$. 
We also considered other values in  $[0.05,0.15)$ as the upper bound, and the estimates of the slope parameter obtained were not significantly different. Also, in each case, the rank was estimated correctly. We do not report these results here. Finally, we also mention that since our procedure involves a computationally intensive minimization procedure to select the rank, we have not iterated the simulation procedure several times to provide MSEs of the estimators, as is the convention. However, we ran the above simulation for a few iterations, and the results were very similar to those provided here. The same was found to be true for all simulations done later in the paper.

\begin{figure}[t]
\vspace{-0.2in}
\begin{center}
{
\includegraphics[scale=0.6]{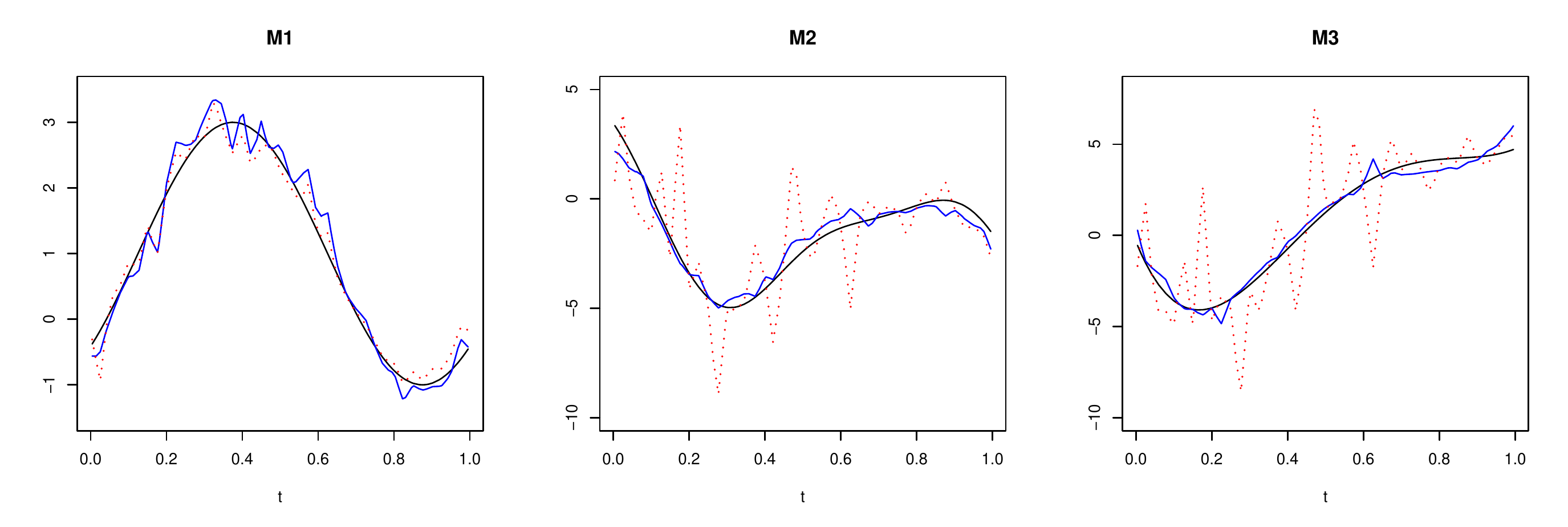}
}
\end{center}
\vspace{-0.3in}
\caption{Plots from the left to the right show the true slope function (black curves), the functional regression calibration estimator (blue curves), and the spectral truncation estimator based on erroneous observations (red dotted curves) under models (M1), (M2) and (M3), respectively, and $\delta = 0.05$.}
\label{Fig1}
\end{figure}

\begin{figure}[t]
\vspace{-0.1in}
\begin{center}
{
\includegraphics[scale=0.6]{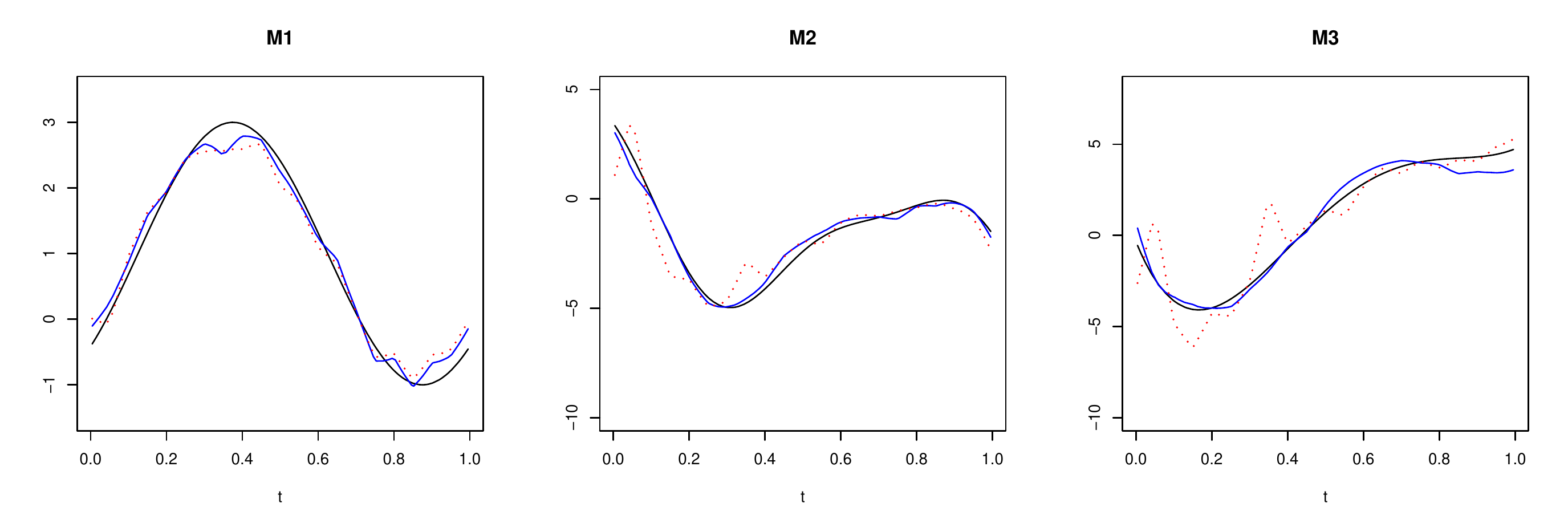}
}
\end{center}
\vspace{-0.3in}
\caption{Plots from the left to the right show the true slope function (black curves), the functional regression calibration estimator (blue curves), and the spectral truncation estimator based on erroneous observations (red dotted curves) under models (M1), (M2) and (M3), respectively, and $\delta = 0.1$.}
\label{Fig2}
\end{figure}

\indent The plots of the true slope functions, the functional regression calibration estimates and the spectral truncation estimates based on erroneous observations are shown in Figures \ref{Fig1} and \ref{Fig2} for $\delta = 0.05$ and $\delta = 0.1$, respectively. Figures \ref{Fig1} and \ref{Fig2} show that the regression calibration estimator adequately captures the true slope parameter for all the models and clearly shows the effects of not accounting for the measurement error in the covariate. The spectral estimator performs poorly under models (M2) and (M3). It is seen that the spectral truncation estimates are smoother when $\delta = 0.1$ compared to $\delta = 0.05$. This is to be expected because of the followin reason. The process $Z$ with $\delta = 0.05$ has more terms in its Karhunen-Lo\`eve expansion compared to $\delta = 0.1$ resulting in slower decay of eigenvalues. This slower decay negatively affects the prediction error (see, e.g., \cite{YC10}) that was used to estimate the spectral cut-off, and thus results in the selection of more components in the spectral estimator. Indeed, for models (M1), (M2) and (M3), the number of components retained are $3$, $8$ and $8$, respectively, when $\delta = 0.05$, and $3$, $6$ and $7$, respectively, when $\delta = 0.1$. Since the eigenfunctions of the measurement error component $U$ has much more sharper spikes for $\delta = 0.05$ compared to when $\delta = 0.1$, these additional components give rise to spiky artefacts in the estimate of the slope parameter. These artefacts are far less pronounced in the regression calibration estimator since it corrects for the measurement error component. This is clearly seen from the plots in Figures \ref{Fig1} and \ref{Fig2}. 

\subsection{Case when the errors are truly i.i.d.}  \label{iid-errors}

In this section, we will consider the typical formulation of measurement error in the functional data analysis literature, where they are assumed to be i.i.d. over the observation grid. In this setup, we will compare the regression calibration estimator with the spectral truncation estimator by \cite{HH07} as in the previous section as well as the PACE estimator in \cite{yao2005} and the SIMEX estimator in \cite{Cai2015}. The last two estimators are designed specifically for the i.i.d. measurement error case. \\
\indent The models and simulation settings are the same as in the previous section with the exception that now we take  the measurement errors to be i.i.d. $N(0,0.25)$ variables. The plots of the true slope functions along with the estimators of the slope functions using the different methods are shown in Figure \ref{Fig-iid-errors}. It is observed that for each of the three models, the regression calibration estimator estimates the true slope function very well as is comparable to the PACE estimator which is specifically designed for i.i.d. errors. Further, under all models, the rank of the true covariate $X$ is correctly estimated using the algorithm in subsection \ref{subsec2-1}. The SIMEX estimator, which is implemented using the linear extrapolation procedure, fails under all three models. In fact, the performance of the linear extrapolation procedure was found to be the best among the linear, the non-linear and the local polynomials extrapolation procedures (see pp. 41--45 in \cite{Cai2015}). The spectral truncation approach fails dramatically under all the three models. This is unlike what we observed in the previous section. This indicates that the behaviour of the estimator of the slope (using erroneous covariates) differs according to the distribution of the error and is thus in accordance with what is known in the multivariate literature (see, e.g., p. 41 in \cite{CRSC06}). The above observations show that although the regression calibration procedure is designed for the case when the measurement error is a valid stochastic process with a $\delta$-banded covariance structure, its performance is unaffected even in the limiting case $\delta = 0$ (when the resulting error structure is no longer a valid stochastic process). A possible reason for this is that although the population covariance structure of the covariate has a discontinuity only at the diagonal (due to the presence of i.i.d. measurement errors), the empirical covariance will not have this behaviour (since the off-diagonal empirical correlations between the errors need not be exactly zero in finite samples). Thus, it is better to remove a small band along the diagonal to recover the underlying true covariance.
\begin{figure}[t]
\vspace{-0.2in}
\begin{center}
{
\includegraphics[scale=0.6]{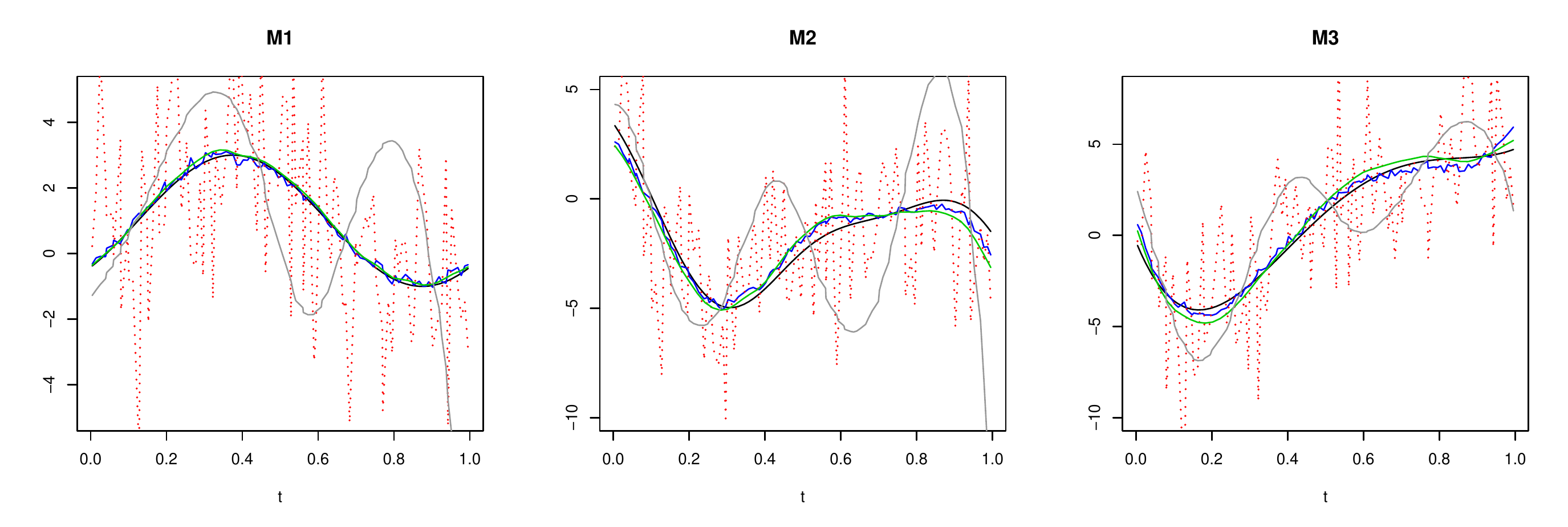}
}
\end{center}
\vspace{-0.3in}
\caption{Plots from the left to the right show the true slope function (black curves), the functional regression calibration estimator (blue curves), the spectral truncation estimator based on erroneous observations (red dotted curves), the PACE estimator (green curves) and the SIMEX estimator (grey curves) under models (M1), (M2) and (M3), respectively, and i.i.d. errors.}
\label{Fig-iid-errors}
\end{figure}

\section{Case when $r = \infty$ but $X$ is essentially finite rank}  \label{sec4}

The methodology and the theory developed thus far in the paper assume that the covariance of $X$ is exactly of finite rank. It may sometimes be the case that $X$ is truly infinite dimensional although it is essentially finite rank (finite rank covariances are dense among trace-class covariances, so any functional datum is essentially of finite rank). Specifically, $X$ has an infinite Karhunen-Lo\`eve expansion, i.e, $r=\infty$, but the eigenvalues decay very fast so that the almost all of the variability of $Y$ is explained by the first few eigenvalues.  \\
\indent There are a few problems that come up when $X$ is truly infinite dimensional although it is essentially finite rank. It is well-known in the functional linear regression literature that including higher order eigenfunctions leads to increased variability of the estimator of the slope function in finite samples due to the involvement of the inverses of these eigenvalues in the term $\widehat{\mathscr{K}}_{X}^{-}$. Also, in this case, since the decay rate is quite fast, the instability will be very significant. So, if we use the earlier method for estimating $r$ as given in Section \ref{sec2}, then we end up getting a very good estimate of the covariance operator (due to the consistency of $\widetilde{r}_{L_{*}}$ by Proposition \ref{prop1}), but a terrible estimate of the slope function. To address this problem, we have modified the method of estimating $r$ which will stabilize the estimator of the slope function. \\
\indent Note that the instability of the estimator depends on the degree of ill-posedness of the (finite rank) estimate of the covariance operator. This is measured by its condition number, i.e., the ratio of the largest to the smallest eigenvalue. The idea is to control the instability of the inverse of the estimator of the covariance operator by imposing a restriction on this condition number, which is in itself a spectral truncation approach. This in turn will impose a restriction on the possible values of the rank of the estimator of $K_{X}$. We only consider these ranks in the optimization criterion in \cite{DP16}. We implement this method as follows. 
\begin{itemize}
\item[Step (i):] As in the grid sub-sampling algorithm earlier, fix a suitable $L_{*}$ and for the $b$th iteration, randomly sub-sample a grid from the original one. Compute the empirical covariance $\widehat{K}_{W*}$ of $W$ for that grid, and compute $f_{b,L_{*}}(j) = \min |||P_{L_{*}\delta_{*}} \circ (\widehat{K}_{W*} - \Theta)|||_{F}^{2}$ over all $L_{*} \times L_{*}$ positive definite matrices $\Theta$ of rank $j$ for each $j = 1,2,\ldots,(L_{*}/4 + 1)$. Here $P_{L_{*}\delta_{*}}$ is as defined in Section \ref{subsec2-1}.

\item[Step (ii):] Repeat Step (i) $B$ times, and for each $j = 1,2,\ldots,(L_{*}/4 + 1)$, the median of $\{f_{b,L_{*}}(j), 1 \leq b \leq B\}$ is recorded. We denote these by $\widetilde{f}_{L_{*}}(j)$'s. 

\item[Step (iii):] Using the full grid, compute consistent estimates of the condition numbers of the rank $j$ estimator of the covariance operator of $Y$ for each $j = 1,2,\ldots,(L_{*}/4 + 1)$. Denote these by $\widetilde{a}_{j}$'s.
 
\item[Step (iv):] Estimate the essential rank by $\widehat{r}_{ess} = \max\{j :  \widetilde{f}_{L_{*}}(j) \leq c_{1}, \widetilde{a}_{j} \leq c_{2}\}$ for pre-determined values of $c_{1} > 0$ and $c_{2} > 0$.

\item[Step (v):] Repeat Step 4* in Section \ref{subsec2-1} with $\widehat{r}_{ess}$ instead of $\widetilde{r}_{L_{*}}$ to obtain an estimator $\widehat{\mathscr{K}}_{X,ess}$ of $\mathscr{K}_{X}$.
\end{itemize}
With these notations, the regression calibration estimator is given by $\widehat{\beta}_{rc,ess} = \widehat{\mathscr{K}}_{X,ess}^{-}\widehat{C}_{y,W}$. Clearly, $\widehat{\beta}_{rc,ess}$ will not be consistent for $\beta$ unless $\widehat{r}_{ess}$ converges to $r$. In order to achieve this, one would have to suitably decrease $c_{1}$ and increase $c_{2}$ with the sample size. Also, the rate of decay of $c_{1}$ and growth of $c_{2}$ will depend on the rate of decay of the eigenvalues. However, at this stage, we do not have a precise theory. \\ 
\indent Throughout the simulations, we have chosen $c_{2} = 50$ in accordance to a thumb-rule in multivariate statistics about how large should a condition number be for the matrix to be ill-conditioned (see, e.g., \cite{Hock03}). Also, as in Section \ref{sec3}, we have chosen $c_{1} = 0.01L_{*}^{2}$. We have conducted simulations to investigate the performance of the estimators in this setting and have chosen the following modifications of the models considered earlier for $X$. 
 \\
(M4) $r = 20$, $\eta_{1} \equiv 1, \eta_{2j}(t) = \sqrt{2}\sin(2{\pi}jt), \eta_{2j+1}(t) = \sqrt{2}\cos(2{\pi}jt)$ for $t \in [0,1]$ and $j \geq 2$. \\
(M5) $r = 20$, $\eta_{j}(t)$'s are the normalized functions in the Gram-Schmidt orthogonalization of the functions $f_{1}, \ldots, f_{5}$ in model (M2) along with the functions $f_{j}(t) = t^{j-3}, j=6,7,\ldots,r$, for $t \in [0,1]$. \\
(M6) $r = 20$, $\eta_{j}(t)$'s are the normalized functions in the Gram-Schmidt orthogonalization of the functions $f_{1}, \ldots, f_{5}$ in model (M3) along with the functions $f_{j}(t) = t^{j-1}, j=6,7,\ldots,r$, for $t \in [0,1]$. \\
The value $r = 20$ is chosen to mimic the potentially infinite rank case. For model (M4), $\lambda_{1}$ to $\lambda_{3}$ was equally spaces between $1.5$ and $0.3$ similar to model (M1), and $\lambda_{j} = {0.1421}(j-3)^{-4}, j=4,5,\ldots,r$. For models (M5) and (M6), $\lambda_{1}$ to $\lambda_{5}$ was equally spaces between $1.5$ and $0.3$ similar to models (M2) and (M3), and $\lambda_{j} = {0.2368}(j-5)^{-4}, j=6,7,\ldots,r$. The choices of the remaining $\lambda_{j}$'s are made to ensure that the first three (respectively, five) eigenvalues in model (M4) (respectively, (M5) and (M6)) explain $95\%$ to the total variation of $X$. So, for model (M4), the essential rank can be taken to be $3$, while it can be taken to be $5$ for models (M5) and (M6). For each of these models, we have considered only one model for $U$, namely, that corresponding to $\delta = 0.1$ in the earlier simulation. The slope functions chosen are $\beta = \eta_{1} + \eta_{2} - \eta_{3} - \eta_{4} + 0.5\eta_{5}$ for model (M4), $\beta = -0.4\eta_{1} + 2\eta_{2} - \eta_{3} + \eta_{4} - 0.7\eta_{5} + 0.5\eta_{6} - 0.3\eta_{7}$ for model (M5), and $\beta = 0.7\eta_{1} + 3\eta_{2} - \eta_{4} + 0.5\eta_{5} + 0.3\eta_{7}$ for model (M6). These choices of the slope parameter have been made so that the highest order eigenfunction in it corresponds to an eigenvalue that is smaller than $0.01$ times the largest eigenvalue, i.e., the condition number exceeds $100$. This will enable us to investigate whether the regression calibration approach described in Steps (i)-(v) above can adapt to the general situation or not. We did not include other higher order eigenfunctions in the choice of $\beta$ because those components have almost negligible contribution to the variability of $X$. So, including them would essentially amount to adding terms which are almost orthogonal to the true covariate $X$, and hence would result in a problem similar to the identifiability issue in functional linear regression when the slope has components orthogonal to the covariate. The sample size and the distribution of the error component in the linear regression are the same as earlier. \\
\begin{figure}[!t]
\vspace{-0.2in}
\begin{center}
{
\includegraphics[scale=0.6]{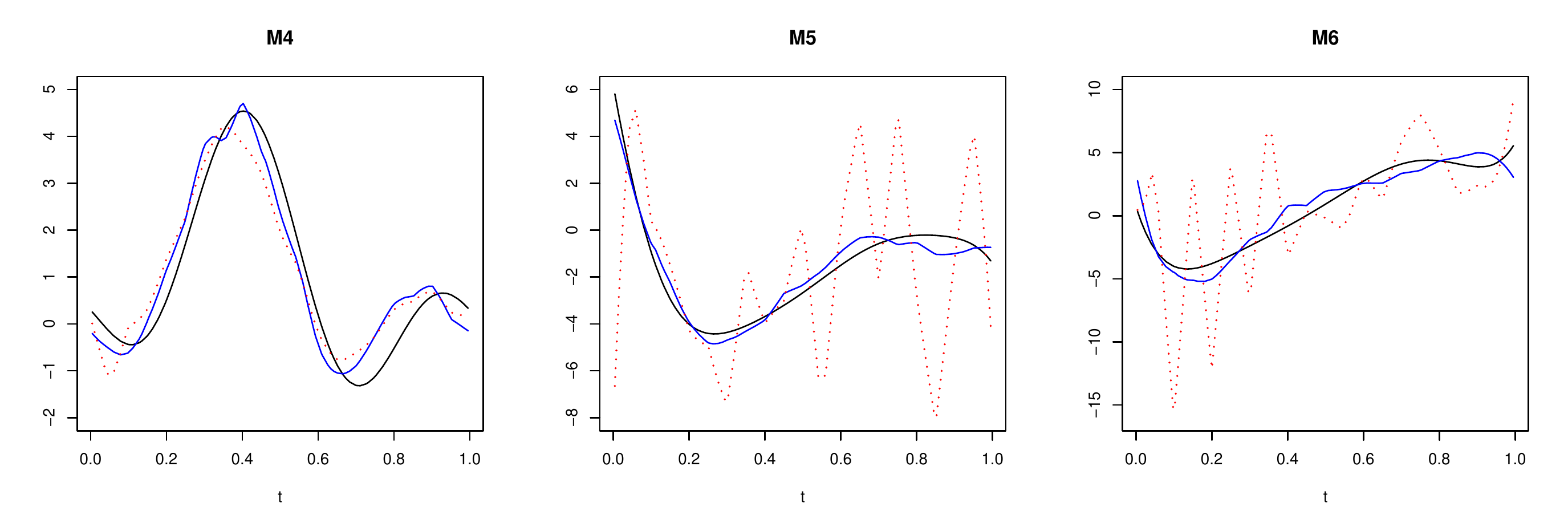}
}
\end{center}
\vspace{-0.3in}
\caption{Plots from the left to the right show the true slope function (black curve), the functional regression calibration estimator (blue curve), and the spectral truncation estimator based on erroneous observations (red dotted curves) under models (M4), (M5) and (M6), respectively, and $\delta = 0.1$.}
\label{Fig3}
\end{figure}
\indent The plots of the true slope functions, the regression calibration estimates and the spectral truncation estimates based on erroneous observations are shown in Figure \ref{Fig3}. It is seen that even in these potentially infinite rank situations, the regression calibration estimator is a reasonably good estimator of the true slope function. The only loss is that for the given sample size, the estimates of the essential rank are $4$, $6$ and $6$ for models (M4), (M5) and (M6), respectively. So, the highest order eigenfunction in the true slope parameters for each of these model is missed. But this in turn stabilizes the estimator and demonstrates the adaptability of the estimator even in the presence of higher order eigenfunctions in the slope parameter. On the other hand, the spectral truncation estimator completely fails for models (M5) and (M6). This is because the spectral cut-off has been grossly over-estimated for each of these models, which implies that higher order eigenfunctions make the estimator highly unstable.

\section{Data Analysis} \label{sec5}

We demonstrate the performance of the regression calibration estimator on the 'gait' dataset that can be obtained in the \textit{fda} package in the \texttt{R} software (see also \cite{RS05}). The dataset contains the hip angles and the knee angles of $39$ boys, and the observations are taken over $20$ different time points during a gait cycle. The aim of the analysis is to fit a function-on-function linear regression with the response as the knee angle and the covariate as the hip angle. Such data typically have measurement errors in both the response and the covariate. We will assume that the error in the covariate is of the white noise type. So, to correct for the same, we apply the regression calibration method by assuming that the bandwidth of the error process is approximately zero. As in the simulations in Section \ref{sec3}, we also compute the spectral truncation estimator by choosing the spectral cut-off in the same way as described in that section. The estimation procedure for the regression calibration estimator is chosen to be the one for the potentially infinite rank scenario described in Section \ref{sec4}. We obtain the essential rank to be $5$, and the spectral cut-off is found to be $4$. The plots of the estimators of the slope parameter using the two estimators is given in Figure \ref{Fig5}. It is observed that the regression calibration approach has `sharpened' the estimate, i.e., the areas of positive and negative impact are much more prominent than in the spectral truncation approach. The plots of the observed $W_{i}$'s, the estimated $X_{i}$'s, and the estimated pointwise variances of the measurement error process are shown in Figure \ref{Fig6}. The latter estimates are obtained using the method in \cite{DP16}. It may be inferred from the plots of the marginal error variances that making an assumption of homoscedastic errors is not justified for this dataset. To further measure the usefulness of the estimates obtained, we used the estimates of the $X_{i}$'s as the covariates and predicted the corresponding knee angle curves using the regression calibration estimator. The $R^{2}$-coefficient of prediction for this case was found to be $54.2\%$, while the same using the spectral truncation estimator and the observed $X_{i}$'s as covariates was found to be $50.3\%$. This indicates a marked  improvement in the fit of the linear regression by correcting for the measurement error in the covariate. \\
\begin{figure}[t]
\begin{center}
{
\includegraphics[trim = 3cm 8cm 2.8cm 8cm, scale=0.4]{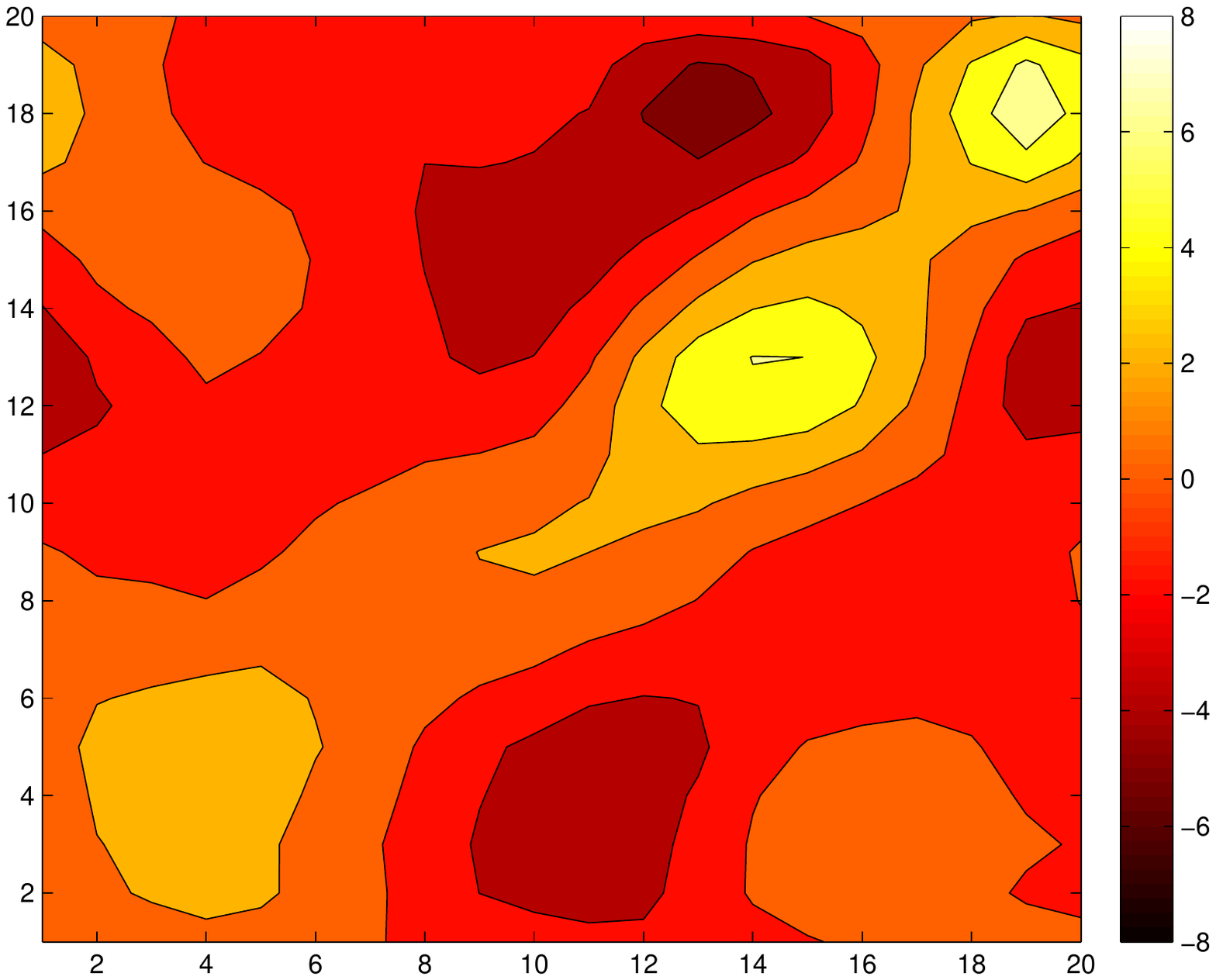}
\qquad
\includegraphics[trim = 3cm 8cm 2.8cm 8cm, scale=0.4]{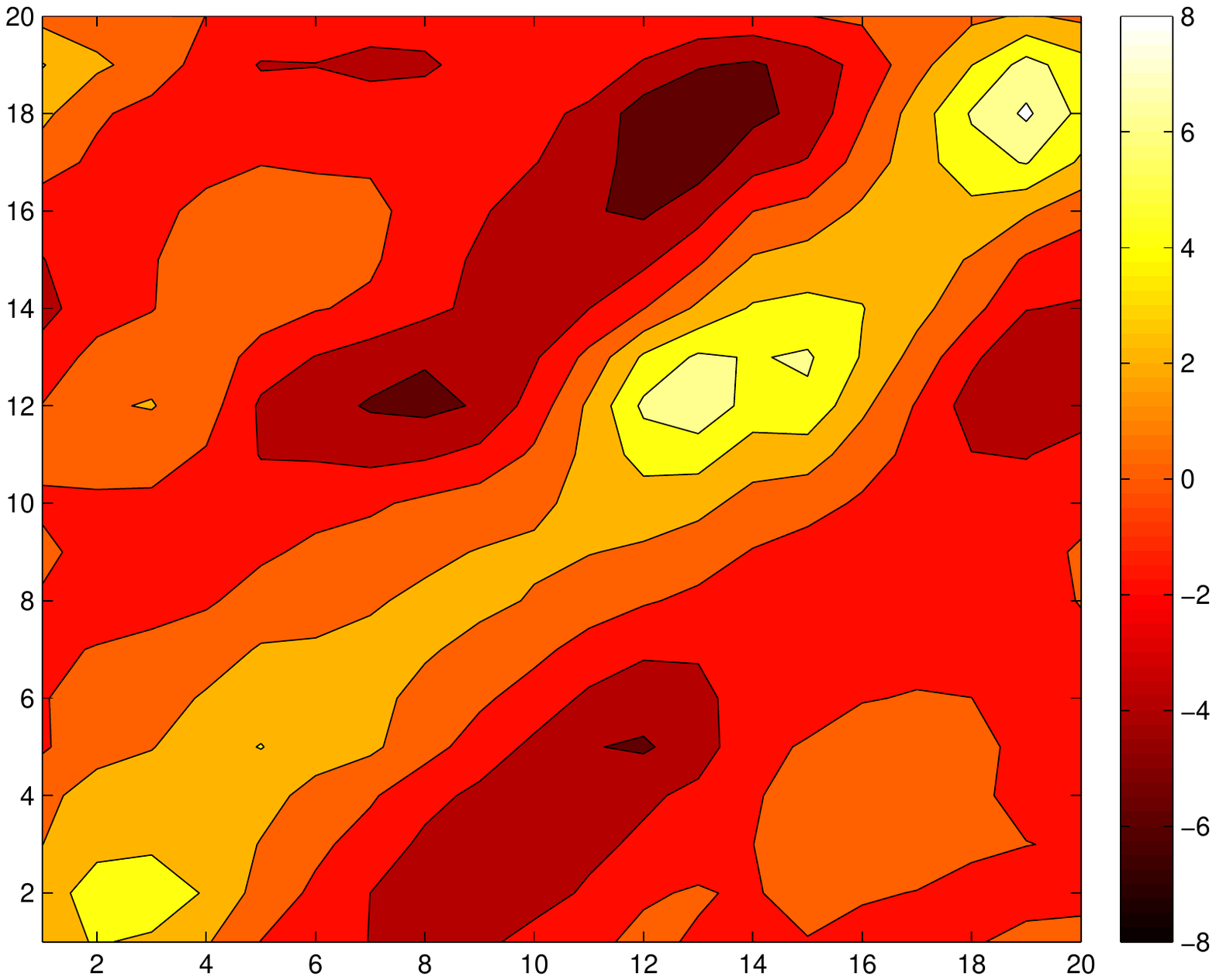}
}
\end{center}
\caption{Plots the spectral truncation estimate (left) and the regression calibration estimate (right) of the slope parameter for the gait data.}
\label{Fig5}
\end{figure}
\begin{figure}[t]
\vspace{-0.1in}
\begin{center}
{
\includegraphics[scale=0.6]{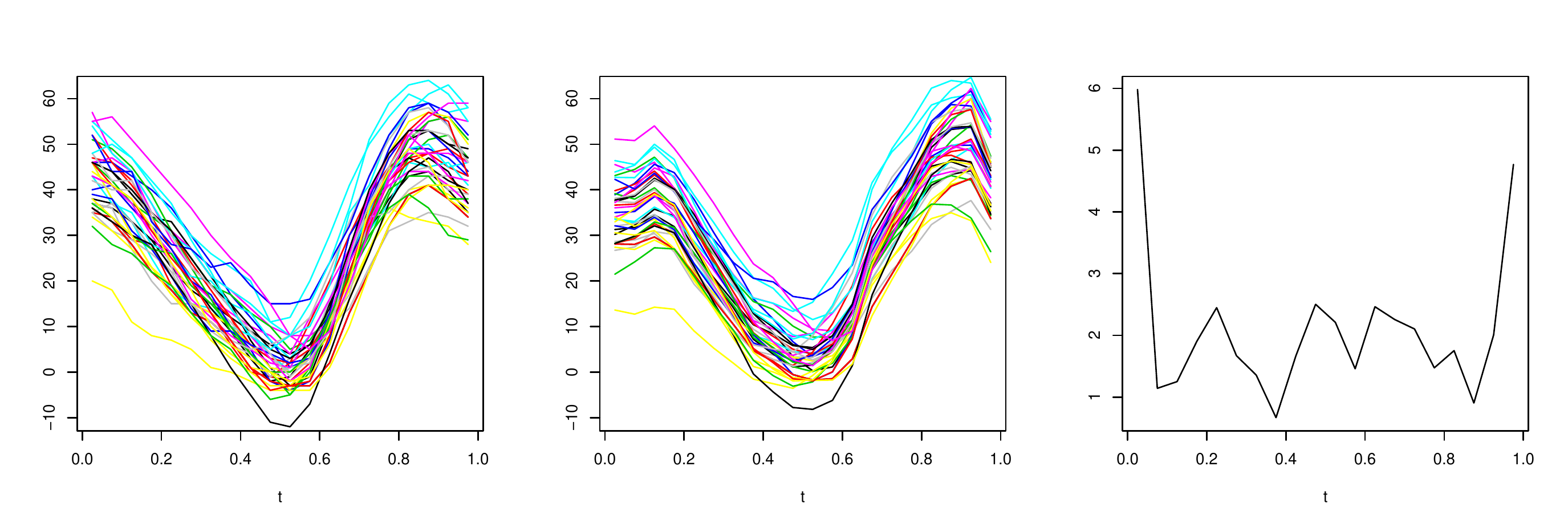}
}
\end{center}
\vspace{-0.3in}
\caption{Plots of the hip angle curves with measurement errors (left), the estimated true hip angle curves corrected for measurement error (middle), and the estimated pointwise variances of the measurement error process (right).}
\label{Fig6}
\end{figure}

\bibliographystyle{apalike}
\bibliography{biblio1}

\end{document}